\documentclass[11pt]{article} 

\usepackage{fullpage}
\usepackage{hyperref}       
\hypersetup{
    colorlinks=true,      
    linkcolor=teal,       
    citecolor=teal,        
    filecolor=magenta,    
    urlcolor=teal         
}

\usepackage{url}            
\usepackage{amsfonts}       
\usepackage{nicefrac}       
\usepackage{graphicx}
\usepackage{natbib}
\usepackage{authblk}
\usepackage[page]{appendix}

\usepackage{amsmath}
\usepackage{amsthm} 
\usepackage{subfigure} 
\usepackage{xcolor}
\usepackage{tikz}  
\usepackage{booktabs} 
\usepackage{caption}  
\usepackage{adjustbox} 
\usepackage[ruled,vlined]{algorithm2e}

\newtheorem{theorem}{Theorem}
\newtheorem{lemma}{Lemma}

\newtheorem{corollary}{Corollary}

\newtheorem{example}{Example}
\newtheorem{remark}{Remark}

\newtheorem{assumption}{Assumption}

\newcommand{\E}{\mathbb{E}}

\allowdisplaybreaks

\title{\bf \LARGE  Causal Data Fusion for Panel Data without a Pre-Intervention Period\\~\\} 

\author[1]{Zou Yang}
\author[1]{Seung Hee Lee}
\author[2]{Julia R. Köhler}
\author[3]{AmirEmad Ghassami\thanks{Corresponding author (\texttt{ghassami@bu.edu)}}}

\affil[1]{{\normalsize Faculty of Computing and Data Sciences, Boston University}}
\affil[2]{{\normalsize Division of Infectious Diseases, Boston Children's Hospital \& Harvard Medical School}}
\affil[3]{{\normalsize Department of Mathematics and Statistics, Boston University}}

\date{First Version: October 21, 2024; Current Version: February 6, 2026}

\begin{document}

\maketitle

\begin{abstract}
Traditional panel-data causal inference frameworks, such as difference-in-differences and synthetic control methods, rely on pre-intervention data to estimate counterfactual means. However, such data may be unavailable in real-world settings when interventions are implemented in response to sudden events, such as public health crises or epidemiological shocks. In this paper, we introduce two data-fusion methods for causal inference from panel data in scenarios where pre-intervention data are unavailable. These methods leverage auxiliary reference domains with related panel data to estimate causal effects in the target domain, thereby overcoming the limitations imposed by the absence of pre-intervention data. We demonstrate the efficacy of these methods by deriving bounds on the absolute bias that converge to zero under suitable conditions, as well as through simulations across a variety of panel-data settings. Our proposed methodology renders causal inference feasible in urgent and data-constrained environments where the assumptions of existing causal inference frameworks are not met. As an application of our methodology, we evaluate the effect of a community organization vaccination intervention in Chelsea, Massachusetts on COVID-19 vaccination rates.\\

\noindent
{\bf Keywords:} Causal Inference; Data Fusion; Panel Data; Equi-Confounding; Synthetic Control
\end{abstract}


\section{Introduction}
\label{sec:intro} 

In many real-world applications, a treatment or intervention is applied to a target unit, and the researcher seeks to assess the causal effect of this intervention by evaluating the counterfactual behavior of the target unit in the absence of the intervention. This setting is challenging because, unlike conventional causal inference scenarios, there is only one treated unit, and there may be unobserved variables in the system that confound the treatment–outcome relationship.

The difference-in-differences framework \citep{card1990impact,angrist2008mostly} and its major extension, the synthetic control framework \citep{abadie2003economic,abadie2010synthetic,abadie2015comparative}, are well-known panel-data methods commonly used in such settings. These methods require access to pre-intervention data from the target unit to match it with units (or combinations of units) from the control group. However, in some applications, such pre-intervention data may not be available. For example, in the case of a sudden shock or emergency, immediate action may be taken by agencies, leaving no data from the unit during the emergency before the action was implemented. During the COVID-19 pandemic, interventions were implemented almost instantaneously to curb the spread of the virus. When COVID-19 vaccines first became available in Massachusetts, doses were accessible primarily to English-fluent people who could use online appointment scheduling sites and who could drive to mass vaccination centers. For many low-income residents who lacked access to information in their preferred language, to a computer, or to a car, COVID-19 vaccines remained out of reach. Efforts by a community-based organization in the city of Chelsea became exemplary for improving vaccine access for their community members. As the organization began advocating for vaccine access and mounting a campaign to provide vaccine information for their community members (which is the intervention in the system) as soon as vaccines became available to the general public in January 2021, these circumstances left researchers with no pre-intervention data from the target unit to analyze the impact of the interventions. This lack of pre-intervention data hampers the ability to draw robust conclusions about the effectiveness of the interventions and to inform future policy decisions.

To address this challenge, in this paper, we propose two data-fusion–based causal inference methods for panel data with no pre-intervention period. We consider settings where, in addition to the target domain in which we aim to assess the causal effect, we have access to panel data from the units in an auxiliary reference domain. The reference domain can correspond to time periods before, concurrent with, or after the time periods in the target domain. Moreover, both the observed and unobserved covariates in the reference domain may differ from those in the target domain. Our first method, \emph{equi-confounding data fusion}, extends the ideas of difference-in-differences, and our second method, \emph{synthetic control data fusion}, extends the synthetic control method. As noted above, both the standard difference-in-differences and synthetic control methods require pre-intervention data. Our contribution is to formalize sufficient conditions on the relationship between the data in the reference and target domains that render identification and estimation of the causal effect feasible via integration (fusion) of data from the two domains. These conditions essentially allow us to anchor the counterfactual outcome of the unit of interest in the target domain to that in the reference domain. 

A challenge introduced by having more than one domain is that extra caution is required when matching the target unit to control units, as naive extensions of existing matching methods may lead to unsatisfactory matches in one or both domains. We address this issue by introducing a novel data-driven constrained optimization method for designing matches in our synthetic control data fusion approach. We demonstrate the efficacy of these methods by deriving bounds on the absolute biases and provide conditions under which these bounds converge to zero. 

The use of data fusion for identification and estimation of causal parameters when data from a single domain are insufficient for effective evaluation has gained significant attention in recent years \citep{dahabreh2020extending,degtiar2023review,colnet2024causal}. However, to the best of our knowledge, we are the first to study this approach in the challenging setting of panel data with a single target unit where latent confounding variables may be present.

The rest of the paper is organized as follows. We provide a formal description of the problem in Section~\ref{sec:desc} and lay out the basic assumptions of the setting under study. We present our equi-confounding data fusion and synthetic control data fusion methods in Sections~\ref{sec:eqcon} and \ref{sec:sc}, respectively. We evaluate our approaches on synthetic data in Section~\ref{sec:sims}. We apply our methodologies to evaluate COVID-19 vaccination rates in Section~\ref{sec:application}. Our concluding remarks are provided in Section~\ref{sec:conc}. All proofs are provided in the Supplementary Material.

\section{Problem Description}
\label{sec:desc}

We consider a setting with $J+1$ units observed over $S$ time periods.
For each unit $i$ at time period $s$, let $A_{i,s}\in\{0,1\}$ and $Y_{i,s}\in\mathbb{R}$ denote the treatment and outcome random variables, respectively, and let $X_i\in\mathbb{R}^{d_t}$ denote the observed covariates for $1\le i\le J+1$. We consider the case where unit $i=1$ receives the treatment over all time periods, and the remaining $J$ units are control units, i.e., they never receive treatment. Formally,
\[
A_{i,s}=
\begin{cases}
	1 & \text{if } i=1, \; 1\le s\le S,\\
	0 & \text{if } 2\le i\le J+1, \; 1\le s\le S.
\end{cases}
\]
We refer to unit $i=1$ as the \emph{target unit}.
An example of this setting is the study of the effect of community organizations on the COVID-19 vaccination rate of a sub-population in cities, as mentioned in Section~\ref{sec:intro}. Here, the target city (Chelsea) had a strong community organization with a large team of community health workers and the capacity to provide wraparound services addressing fundamental social determinants of health from the beginning of the vaccine rollout, whereas the remaining cities under consideration did not have such community organizations. 

We adopt the potential outcomes framework \citep{rubin1974estimating} and denote the potential outcome of $Y_{i,s}$, had the treatment variable been set to $A_{i,s}=a$ for all $s$ (possibly contrary to fact), by $Y_{i,s}^{(a)}$, where $a\in\{0,1\}$. Our parameter of interest is the effect of the treatment on the treated (ETT) averaged over the $S$ treatment time periods:
\[
\psi_0=\frac{1}{S}\sum_{s=1}^S \E\!\left[Y_{i,s}^{(1)}-Y_{i,s}^{(0)} \;\middle|\; A_{i,s}=1 \;\text{for all } s\right].
\]
Since unit $i=1$ is the only unit with $A_{i,s}=1$, the parameter $\psi_0$ can be written as
\begin{equation}
\label{eq:ett}
\psi_0=\frac{1}{S}\sum_{s=1}^S \E\!\left[Y_{1,s}^{(1)}-Y_{1,s}^{(0)}\right].
\end{equation}

We assume the consistency assumption \citep{hernan2020causal}, formally stated as follows.
\begin{assumption}[Consistency]
\label{assm:cons}
	For all $1\le i\le J+1$ and $1\le s\le S$, we have $Y_{i,s}=Y_{i,s}^{(a)}$ if $A_{i,s}=a$ for all $s$. 
\end{assumption}

By consistency, $\psi_0$ can be written as $\psi_0=\frac{1}{S}\sum_{s=1}^S \E[Y_{1,s}-Y_{1,s}^{(0)}]$. Hence, the challenge in evaluating the causal effect lies in assessing the potential outcome mean of the target unit under control. Similar settings have been considered in the literature on causal inference within the difference-in-differences framework, the synthetic control framework, and their variations \citep{arkhangelsky2024causal}. However, in those approaches it is assumed that the target unit is also observed before the implementation of the treatment, which is not the case in our setting. Therefore, those approaches cannot be directly applied here.

In our setting, we also have access to data from a secondary domain for the same units. We refer to the primary domain, in which we aim to evaluate the causal parameter, as the \emph{target domain}, and to the secondary domain as the \emph{reference domain}. 
In our application, the target domain consists of the COVID-19 vaccination rates of the Hispanic sub-population in Massachusetts cities, while the reference domain is the vaccination rates of the Black sub-population in those same cities.
In the reference domain, units are observed over $T$ time periods. We denote the observed covariates of the units in this domain by $Z_i\in\mathbb{R}^{d_r}$, $1\le i\le J+1$. Note that the covariates observed in the target and reference domains need not be the same, and $d_r$ is not necessarily equal to $d_t$. We also observe an outcome variable $F_{i,t}\in\mathbb{R}$ for $1\le i\le J+1$ and $1\le t\le T$ in the reference domain. 

Our goal is to formalize conditions regarding the relevance of the outcome variables $F_{i,t}$ in the reference domain and $Y_{i,s}$ in the target domain, under which the observations in the reference domain can be leveraged to evaluate the causal parameter of interest in the target domain. To this end, we propose two data-fusion approaches, termed equi-confounding data fusion and synthetic control data fusion, which are described in Sections~\ref{sec:eqcon} and \ref{sec:sc}, respectively.

\section{Equi-Confounding Data Fusion}
\label{sec:eqcon}

\subsection{Linear Equi-Confounding}

Our first proposed approach extends the ideas of negative controls and the difference-in-differences causal inference framework to our data fusion setup. A similar extension appeared in \citep{ghassami2022combining} for evaluating the long-term causal effect of treatments via data fusion. 

Define
\[
Y^{(0)}_i:=\frac{1}{S}\sum_{s=1}^S Y^{(0)}_{i,s}, \quad 
Y_i:=\frac{1}{S}\sum_{s=1}^S Y_{i,s}, \quad 
F_i:=\frac{1}{T}\sum_{t=1}^T F_{i,t}.
\]
The main assumption of our approach is the following.

\begin{assumption}[Linear equi-confounding]
\label{assm:equi}
The potential outcome variables in the target domain and the outcome variables in the reference domain satisfy
\begin{equation}
    \label{eq:equi}
\E[Y_1^{(0)}]
-\frac{1}{J}\sum_{i=2}^{J+1}\E[Y_i^{(0)}]
= 
\E[F_1]
-\frac{1}{J}\sum_{i=2}^{J+1}\E[F_i].
\end{equation}
\end{assumption}

The left-hand side of Equation~\eqref{eq:equi} involves variables in the target domain, while the right-hand side involves variables in the reference domain. Hence, Assumption~\ref{assm:equi} specifies the connection between the outcomes in the two domains and clarifies how information in the reference domain can be transferred to the target domain. The quantity 
\(
\E[Y_1^{(0)}]-\frac{1}{J}\sum_{i=2}^{J+1}\E[Y_i^{(0)}]
\) 
measures the confounding effect on the potential outcome $Y^{(0)}_1$ (in our application, this could be, e.g., the effect of regional socioeconomic barriers or shared healthcare infrastructure). Therefore, Assumption~\ref{assm:equi} requires that this effect for $Y^{(0)}_1$ and $F_1$ be equal; hence the name \emph{equi-confounding}. The assumption is inspired by the additive equi-confounding bias assumption in the identification approach using negative outcome controls \citep{lipsitch2010negative,tchetgen2014control,sofer2016negative}, and is similar in spirit to the parallel trends assumption in the difference-in-differences framework \citep{card1990impact,angrist2008mostly}.

To gain intuition about Assumption~\ref{assm:equi}, suppose the intervention in the target domain has no causal effect on the outcome variables in the reference domain. For instance, in the vaccine rollout application, this may hold if the outcome variable in the reference domain is the vaccination rate for another sub-population—e.g., the Black sub-population—which, due to language barriers, was not targeted by the community organization. Another example is the influenza vaccination rate for the same sub-population in the case that the community organization did not campaign for influenza vaccines. Then, Assumption~\ref{assm:equi} requires that the discrepancy in the mean potential outcomes in the absence of community organization activities in the target domain mirrors that in the reference domain. For example, the difference in the mean COVID-19 vaccination rate for the Black sub-population between the target unit and the average of the remaining units is the same as the corresponding difference for the Hispanic sub-population if the target unit (i.e., the city of Chelsea) did not receive any influence from the community organization.

Below, we present an example of a model in which Assumption \ref{assm:equi} is satisfied.

\begin{example}
\label{example:struc}
Assumption~\ref{assm:equi} is satisfied if the data are generated from the following structural equations:
\[
Y_{i,s} = \varrho_s + \theta G_i + \varphi^\top X_i + \alpha_{i,s}A_{i,s} + \varepsilon_{i,s}, 
\quad
F_{i,t} = \rho_t + \theta G_i + \phi^\top Z_i + \epsilon_{i,t},
\]
with $X_i \overset{i.i.d.}{\sim} P_X$ and $Z_i \overset{i.i.d.}{\sim} P_Z$, for some distributions $P_X$ and $P_Z$,
where $\varepsilon_{i,s}$ and $\epsilon_{i,t}$ are mean-zero independent noise variables, and $G_i$ is a target-unit indicator variable equal to $1$ if $i=1$ and $0$ otherwise.
\end{example}

Consider the following estimator for the causal parameter $\psi_0$:
\[
\hat\psi^{eq1}=
\left\{Y_1-F_1\right\}
-\frac{1}{J}\sum_{i=2}^{J+1}\left\{Y_i- F_i\right\}.
\]
The following result establishes the identifiability of $\psi_0$ in our setting and shows that $\hat\psi^{eq1}$ is an unbiased estimator.
\begin{theorem}
\label{thm:equi}
Under Assumptions~\ref{assm:cons} and \ref{assm:equi}, the causal parameter $\psi_0$ is identified as
\[
\psi_0=
\left\{\E(Y_1)-\E(F_1)\right\}
-\frac{1}{J}\sum_{i=2}^{J+1}\left\{\E(Y_i)-\E(F_i)\right\}.
\]
\end{theorem}
\begin{corollary} 
Under Assumptions~\ref{assm:cons} and \ref{assm:equi}, $\hat\psi^{eq1}$ is an unbiased estimator of $\psi_0$.
\end{corollary}

\subsection{Logarithmic Equi-Confounding}
Suppose outcome variables $F_{i,t}$ are of the same nature as $Y_{i,s}$ but measured on a multiplicative scale. In such cases, the outcomes in the reference domain can still be informative about the potential outcomes in the target domain, but transferring information becomes challenging if the multiplicative scale is unknown. For this scenario, we propose a modification of the linear equi-confounding framework that characterizes the relationship of the variables on a multiplicative scale. 

\begin{assumption}[Logarithmic equi-confounding]
\label{assm:logequi}
The potential outcome variables in the target domain and the outcome variables in the reference domain satisfy
\begin{equation}
\label{eq:logequi}
\frac{ \E[Y_1^{(0)}]}{\E\!\left[\sum_{i=2}^{J+1} Y_i^{(0)}\right]}
=
\frac{ \E[F_1]}{\E\!\left[\sum_{i=2}^{J+1} F_i\right]}.
\end{equation}
\end{assumption}

Similar to Equation~\eqref{eq:equi}, in Equation~\eqref{eq:logequi} the left-hand side involves variables in the target domain and the right-hand side involves variables in the reference domain. Hence, Assumption~\ref{assm:logequi} specifies the connection between outcomes in the two domains. This assumption can be viewed as a modification of Assumption~\ref{assm:equi}, in which the change in the mean is replaced by the change in the logarithm of the mean; hence the name \emph{logarithmic equi-confounding}. A similar observation has been made in \citep{wooldridge2023simple}. More specifically, Equation~\eqref{eq:logequi} can be written as
\[
\log \left\{\E(Y_1^{(0)})\right\} - \log \left\{\E\!\left(\sum_{i=2}^{J+1} Y_i^{(0)}\right)\right\} 
= \log \left\{\E (F_1) \right\} - \log \left\{\E\!\left(\sum_{i=2}^{J+1} F_i\right) \right\}.
\]

We have the following identification result.

\begin{theorem}
\label{thm:logequi}
Under Assumptions~\ref{assm:cons} and \ref{assm:logequi}, the causal parameter $\psi_0$ is identified as
\[
\psi_0=
\E[Y_1] - \frac{ \E[F_1]}{\E\!\left[\sum_{i=2}^{J+1} F_i\right]} \E\!\left[\sum_{i=2}^{J+1} Y_i\right].
\]
\end{theorem}

Consider the following estimator for the causal parameter $\psi_0$:
\[
\hat\psi^{eq2}=
Y_1 - \frac{F_1}{ \sum_{i=2}^{J+1} F_i} \sum_{i=2}^{J+1} Y_i.
\]
It can be shown that, under the assumptions of Theorem~\ref{thm:logequi}, this estimator is consistent for $\psi_0$. However, in our setting we only have one observation from the target unit, and hence such a consistency result is not relevant. Instead, we aim to bound the bias of the estimator. To this end, we impose the following conditions.

\begin{assumption}
\label{assm:bound}
~\\[-6mm]
\begin{enumerate}
    \item \emph{(Boundedness).} For all $i$, $Y_{i}^{(0)} \in [l_y, L_y]$ and $F_i \in [l_f, L_f]$, with $0<l_y,l_f$.
    \item \emph{(Weak dependence among controls).} There exists a function $\tau$ with $\tau(J)\downarrow 0$ as $J\to\infty$, such that
    \[
    \frac{1}{J^2}\sum_{\substack{i,j=2\\ i\neq j}}^{J+1}\big|\mathrm{cov}(F_i,F_j)\big| \;\le\; \tau(J).
    \]
    
\end{enumerate}
\end{assumption}

Part (2) of Assumption~\ref{assm:bound} rules out a fully connected dependence structure among control units. This assumption is trivially satisfied if the outcomes for the control units in the reference domain are pairwise independent. 

Let $C(J)$, denote an upper bound for the quantity
\[
\left|\mathrm{cov}\!\left(\frac{F_1}{\frac{1}{J}\sum_{i=2}^{J+1} F_i}, \frac{1}{J}\sum_{i=2}^{J+1} Y_i\right)\right|.
\]  
We have the following result regarding bounding the bias of our proposed estimator.

\begin{theorem}
\label{thm:logequibound}
Under Assumptions~\ref{assm:cons}, \ref{assm:logequi}, and \ref{assm:bound}, the absolute bias of the estimator $\hat\psi^{eq2}$ can be bounded as
\[
\big|\E[\hat\psi^{eq2}]-\psi_0\big|
\;\le\;
L_y\!\left(
\frac{\Delta_f^2}{4\,l_f^2}\,\frac{1}{\sqrt{J}}
\;+\;
\frac{\Delta_f}{2\,l_f^2}\sqrt{\tau(J)}
\;+\;
\frac{L_f}{l_f^3}\Big(\frac{\Delta_f^2}{4J}+\tau(J)\Big)
\right)
\;+\; C(J),
\]
where $\Delta_f:=L_f-l_f$. 
\end{theorem}

Theorem \ref{thm:logequibound} demonstrates that as $J\to\infty$ the bias will be controlled by $C(J)$. Note that often it is expected that $C(J)$ also vanishes as $J\to\infty$.
In this case, the bound on the bias of our proposed estimator vanishes as the number of control units increases.
Note that this property does not require the outcome variables in the target and reference domains to be (approximately) uncorrelated. Instead, it requires that the sample mean $1/J\sum_{i=2}^{J+1} Y_i$ be approximately uncorrelated with the \emph{ratio quantity} $F_1/(1/J\sum_{i=2}^{J+1} F_i)$, which is a weaker requirement, as the ratio can remain approximately unchanged even when the individual variables vary.


The bound in Theorem~\ref{thm:logequibound} can be further improved if we also bound the cross-covariance between the target and control units:

\begin{corollary}
\label{cor2}
Suppose the assumptions of Theorem~\ref{thm:logequibound} hold and, in addition, there exists a function $\tau_1$ with $\tau_1(J)\downarrow 0$ as $J\to\infty$ such that
\[
\frac{1}{J}\sum_{i=2}^{J+1}\big|\mathrm{cov}(F_1,F_i)\big| \;\le\; \tau_1(J).
\]
Then, the absolute bias of the estimator $\hat\psi^{eq2}$ can be bounded as
\[
\big|\E[\hat\psi^{eq2}]-\psi_0\big|
\ \le\
L_y\!\left(
\frac{\tau_1(J)}{l_f^{2}}
+\frac{\Delta_f+L_f}{l_f^{3}}\Big(\frac{\Delta_f^{2}}{4J}+\tau(J)\Big)
\right)
+ C(J).
\]
\end{corollary}

As shown in this section, the equi-confounding data fusion framework allows for simple identification and estimation procedures. However, the assumptions of this framework may rule out certain data-generating processes with time–unit interaction terms. In the next section, we present an alternative data fusion approach that allows for more flexible data-generating processes.

\section{Synthetic Control Data Fusion}
\label{sec:sc}

The equi-confounding assumptions introduced in the previous section require a relation between the average of the outcomes for the control units and the outcomes of the target unit in the reference and target domains. 
In this averaging, all units are weighted equally, with weights $1/J$. In many applications, certain control units are more similar to the target unit and, hence, may contain more information regarding the counterfactual outcomes of the target unit. Therefore, a natural extension of the ideas in the previous section is to consider different weights for different control units. This is the main idea of the synthetic control method \citep{abadie2003economic,abadie2010synthetic,abadie2015comparative}. Our goal is to formalize this idea in our data fusion setting.

We consider weights $w_i$ for the $i$th control unit, $2\le i\le J+1$, and define
$w=[w_2~\cdots~ w_{J+1}]^\top$, such that $w_i \geq 0$ for $i=2,\ldots,J+1$ and
$\sum_{i=2}^{J+1} w_i = 1$.
Let $\bar F_i := [F_{i,1}~\cdots~F_{i,T}]^\top \in \mathbb{R}^T$.

Our goal is to choose weights such that the following \emph{approximate matching} requirements hold:
there exists a constant $c\ge 0$ such that, with probability one,
\begin{align}
\label{eq:match1}
\left\|\bar F_1 - \sum_{i=2}^{J+1} w_i \bar F_i \right\|_2^2 &\le T c^2,\\
\label{eq:match2}
\left\| Z_1 - \sum_{i=2}^{J+1} w_i Z_i \right\|_2^2 &\le d_r c^2,\\
\label{eq:match3}
\left\| X_1 - \sum_{i=2}^{J+1} w_i X_i \right\|_2^2 &\le d_t c^2.
\end{align}
When $c=0$, these conditions reduce to exact matching (perfect fit).

The intuition for imposing \eqref{eq:match1}--\eqref{eq:match3} is as follows.
Suppose the outcomes of the units in the reference and target domains are generated by mechanisms that differ primarily in how they are affected by observed covariates.
In this case, if we match on observed covariates and match outcomes over a long period in the reference domain, then we should also be matching on unobserved covariates that affect outcomes in both domains.
For instance, in the vaccination-rate application, with the Hispanic and Black sub-populations as the target and reference domains, respectively, factors such as overall attitudes toward vaccines, socioeconomic status, and access to vaccines are unobserved covariates that may similarly affect COVID-19 vaccination rates across both sub-populations.
Therefore, if we match on the vaccination rate from one subpopulation for an extended period, as well as on all observed covariates relevant to COVID-19 vaccination, then there is statistical justification that we are also matching on unobserved covariates.

Based on the above argument, a weighted sum of the potential outcomes with weight vector $w$ that
satisfies \eqref{eq:match1}--\eqref{eq:match3} with a small $c$ will match the potential outcome of
the target unit to a good approximation. That is,
\[
Y_1^{(0)} \approx \sum_{i=2}^{J+1} w_i Y_{i}^{(0)}=\sum_{i=2}^{J+1} w_i Y_{i},
\]
where the last equality follows from the consistency assumption.
We refer to the weighted sum on the right-hand side as the \emph{transferred synthetic control}.
Equipped with this quantity, we propose the following estimator for our causal parameter of interest:
\begin{equation}
\label{eq:sc}
\hat{\psi}^{sc} =  Y_{1} - \sum_{i=2}^{J+1} w_i Y_{i}.
\end{equation}

Consider the case of $c=0$. In restrictions \eqref{eq:match1}–\eqref{eq:match3}, we will usually have more equations than unknowns, and, hence, these restrictions cannot be exactly satisfied. Therefore, we instead minimize the mismatch in each restriction. Formally, for $i = 1, \dots, J+1$, let $K_i := [\bar F_{i}^\top~Z_i^\top~X_i^\top]^\top$. A naive way to obtain the weight vector $w$ is
\begin{equation}
\begin{aligned}
\label{eq:optim1}
w =\arg\min_{w} \left\| K_1 - \sum_{i=2}^{J+1} w_i K_i \right\|_2^2,
\end{aligned}
\end{equation}
subject to
$\sum_{i=2}^{J+1} w_i = 1,\quad w_i \geq 0,\ i = 2, \ldots, J+1$.
This is akin to the methodology used in the original synthetic control method. However, this basic approach does not suffice for our data fusion goal.
An important concern with the approach above is that if $T$ is larger than $d_t$, the chosen weight vector $w$ may be mostly influenced by $\bar F_i$, resulting in good matching in the reference domain but poor tracking of $X_i$ in the target domain, and hence poor matching of potential outcomes in the target domain. 
Note that this issue does not arise in the standard synthetic control setting because, in that setting, matching on the outcomes helps matching on the observed covariates and vice versa. In our data fusion setting, matching on observed variables in one domain does not necessarily improve matching in the other domain. 
To quantify the performance of the matchings, for a given $w$, we define the normalized squared errors (NSEs) of the predictors of $\bar F_1$, $Z_1$, and $X_1$ as 
\[
NSE(O_1, w) := \frac{1}{|O_1|}\left\|O_1 - \sum_{i=2}^{J+1} w_i O_i\right\|_2^2, 
\]
where $O_i$ denotes $\bar F_i$, $Z_i$, or $X_i$. 

To resolve the issue above, we add an extra constraint to the optimization problem that forces the matching on the observed covariates to be close to the best achievable matching. We further incorporate a \emph{budgeting} vector $b$ into our objective function to achieve more flexibility in the choice of the weights. Formally, let $b=[b_F~b_Z~b_X]^\top$ be a vector with positive elements that sum to one, determining 
the budgets to spend on minimizing the NSEs for $\bar F_1$, $Z_1$, and $X_1$. 
This vector can be simply set to $b=[1/3~1/3~1/3]^\top$. Alternatively, we can prioritize minimizing the NSE of $\bar F_1$, $Z_1$, or $X_1$—for example, if one is known to be more complex—by increasing its corresponding element in $b$. In general, this can be done in a data-driven fashion using estimated NSEs. We next propose an algorithm that chooses the budget vector $b$ in a data-driven approach and uses it to find an optimal weight vector $w^*$.

Our method for choosing the optimal budget vector $b^*$ and weights $w^*$ is outlined in Algorithm~\ref{data-driven}. This algorithm takes the observed data, as well as two hyperparameters $\eta_Z$ and $\eta_X$, as input. We first calculate baseline weights $w^*_Z$ and $w^*_X$ separately by minimizing $NSE(Z_1, w)$ and $NSE(X_1,w)$, respectively. These baseline weights serve as reference points for constraining the optimization process, as they yield the best possible NSEs for $X_1$ and $Z_1$. 
Next, we define $B$ as the set of all possible vectors $b$ satisfying $b_F, b_Z, b_X \geq 0$ and $b_F + b_Z + b_X= 1$.
For each vector $b\in B$, the algorithm finds $w^*(b)$ by solving optimization problem~\eqref{optim:algo}. This optimization problem aims to find $w$ that minimizes the weighted sum of NSEs. Importantly, we require the solution to be such that the NSEs for $Z_1$ and $X_1$ are not far from the best possible NSEs for these variables (using baseline weights $w^*_Z$ and $w^*_X$). This ensures that the matching for the covariates is not compromised and, hence, resolves the concern mentioned earlier.
Finally, the algorithm picks $b^*$ by minimizing $NSE(\bar F_1, w^*(b))$, and returns $w^*=w^*(b^*)$. Therefore, among all choices with acceptable matches on the covariates, we pick the one with the best match on the outcomes in the reference domain. This choice promotes a good match on unobserved components involved in the outcome-generating mechanism.

\begin{algorithm}[t]
\textbf{Input:} 
$(\bar F_i, Z_i, X_i)_{i=1}^{J+1}$, and hyperparameters $\eta_Z$, $\eta_X$. 

\textbf{Define:} 
$B = \{b = (b_F, b_Z, b_X) : b_F, b_Z, b_X \geq 0,\ b_F + b_Z + b_X = 1 \}$.

\textbf{Calculate:}\\ 
$w^*_Z = \arg \min_w NSE(Z_1, w)$ s.t. $\sum_i w_i = 1,\ 0 \leq w_i$; \\
$w^*_X = \arg \min_w NSE(X_1, w)$ s.t. $\sum_i w_i = 1,\ 0 \leq w_i$.

\For{$b \in B$}{

    \hspace{5mm}Solve:
    \begin{equation}
    \label{optim:algo}
    \begin{aligned}
        w^*(b) &= \arg \min_w\ \ b_F \cdot NSE(\bar F_1, w) + b_Z \cdot NSE(Z_1, w) + b_X \cdot NSE(X_1, w)  \\
            \text{s.t.}\ &\ \ \frac{1+NSE(Z_1, w)}{1+NSE(Z_1, w^*_Z)} \leq 1+ \eta_Z,\quad 
            \frac{1 + NSE(X_1, w)}{1 + NSE(X_1, w^*_X)} \leq 1 + \eta_X, \\
            &\ \ \sum_i w_i = 1,\quad 0 \leq w_i.
    \end{aligned}
    \end{equation}
}

\textbf{Set }$b^* = \arg \min_{b \in B} NSE(\bar F_1, w^*(b))$.

\textbf{Return }$w^*=w^*(b^*)$. 

\caption{Synthetic control data fusion algorithm}
\label{data-driven}
\end{algorithm}

\begin{remark}
In the approach above, in line with the original synthetic control method, we require the weights $w$ to be nonnegative and to sum to one. The purpose of this restriction is interpretability and avoiding overfitting. Yet, similar to the standard synthetic control method, we can also prevent overfitting by regularization \citep{doudchenko2016balancing}.
\end{remark}

As mentioned earlier, our proposed method for approximating the counterfactual outcome of the target unit in the target domain is expected to perform well if the outcomes in the reference and target domains are generated by mechanisms that differ primarily in how they are affected by the observed covariates. We formalize this intuition by considering the following structural equations for the outcome variables in the reference and target domains to study the effect of unobserved confounders on the bias of the estimator.

\begin{assumption}
\label{assm:factormodel}
The outcome variables in the reference and target domains are generated from the following factor models. For all $1\le i\le J+1$, $1\le s\le S$, and $1\le t\le T$,
\begin{equation}
\label{eq:factor}
\begin{aligned}
    Y_{i,s} &= \varrho_s + \varphi_s^\top X_i + \vartheta_s^\top\mu_i 
   + \tilde\vartheta_s^{\top} \tilde\mu_i^Y
    + \alpha_{i,s}A_{i,s} + \varepsilon_{i,s},\\
    F_{i,t} &= \rho_t + \phi_t^\top Z_i + \theta_t^\top\mu_i + \tilde\theta_t^{\top} \tilde\mu_i^F+ \epsilon_{i,t}.
\end{aligned}    
\end{equation}
Here, $\alpha_{i,s}$ is the causal effect on unit $i$ in time period $s$. The vectors $\mu_i\in\mathbb{R}^{d_u}$ are unit-specific latent factors that affect both domains; $\tilde\mu_i^Y\in\mathbb{R}^{d^Y_{\tilde u}}$ and $\tilde\mu_i^F\in\mathbb{R}^{d^F_{\tilde u}}$
are unit-specific latent factors that affect only the target and reference domains, respectively.
Loadings are deterministically bounded in $\ell_2$-norm:
$\sup_{t\le T}\|\theta_t\|_2\le \bar\theta$,
$\sup_{t\le T}\|\tilde\theta_t\|_2\le \bar{\tilde\theta}$,
$\sup_{s\le S}\|\vartheta_s\|_2\le \bar\vartheta$,
and $\sup_{s\le S}\|\tilde\vartheta_s\|_2\le \bar{\tilde\vartheta}$.
We also assume the observed-covariate loadings are deterministically bounded:
$\sup_{s\le S}\|\varphi_s\|_2\le \bar\varphi$ and $\sup_{t\le T}\|\phi_t\|_2\le \bar\phi$.
The matrix $\sum_{t=1}^{T}\theta_t\theta_t^\top$ has minimum eigenvalue $\lambda_{\min}\ge T\,\underline\xi>0$.
The shocks satisfy $\mathbb{E}[\epsilon_{i,t}]=\mathbb{E}[\varepsilon_{i,s}]=0$, and $\epsilon_{i,t}$ are i.i.d.\ sub-Gaussian with variance proxy $\bar\sigma^2$.
The collections $\{\tilde\mu^Y_i\}_{i\le J+1}$ and $\{\tilde\mu^F_i\}_{i\le J+1}$ are i.i.d., mean-zero, sub-Gaussian with variance proxy $\tau^2$.
Let $\mathcal{G}:=\sigma\!\big(\{X_i,Z_i,\mu_i,\tilde\mu^F_i\}_{i\le J+1},\{\epsilon_{i,t}\}_{i\le J+1,\,1\le t\le T}\big)$ denote the $\sigma$-field generated by all reference-domain objects used to construct the synthetic-control weights $w=(w_2,\ldots,w_{J+1})$, which are assumed $\mathcal{G}$-measurable, nonnegative, and satisfy $\sum_{i=2}^{J+1} w_i=1$.
We assume that $\{\varepsilon_{i,s}\}_{i,s}$ and $\{\tilde\mu^Y_i\}_i$ are independent of $\mathcal{G}$.
\end{assumption}

In Assumption~\ref{assm:factormodel} we allow for three groups of unobserved unit-specific factors. The first group, $\mu_i$, is shared across the two factor models, i.e., across the two domains. For example, in the vaccination-rate application mentioned earlier, $\mu_i$ may encode the overall attitude and beliefs regarding vaccination for the $i$th individual. The other two groups are additional unobserved unit-specific factors: $\tilde\mu_i^Y$ in the target domain and $\tilde\mu_i^F$ in the reference domain.
Under the factor model \eqref{eq:factor}, the causal parameter of interest can be written as
\[
\psi_0 \;=\; \frac{1}{S}\sum_{s=1}^{S}\alpha_{1,s}.
\]

We have the following result regarding the bias of our proposed estimator.

\begin{theorem}
\label{thm:syn}
Suppose Conditions \eqref{eq:match1}--\eqref{eq:match3}, Assumption~\ref{assm:cons},
and Assumption~\ref{assm:factormodel} hold. Then
\[
\big|\mathbb{E}[\hat\psi^{sc}]-\psi_0\big|
\;\le\;
\Bigg(
\bar{\varphi}\sqrt{d_t}
\;+\;
\frac{\bar\vartheta\,\bar\theta}{\underline\xi}
\;+\;
\frac{\bar\vartheta\,\bar\theta\,\bar{\phi}}{\underline\xi}\sqrt{d_r}
\Bigg)c
\;+\;
\sqrt{2}\,\frac{\bar\vartheta\,\bar\theta}{\underline\xi}\,\bar\sigma
\sqrt{\frac{\log(2J)}{T}}
\;+\;
2\,\frac{\bar\vartheta\,\bar\theta\,\bar{\tilde\theta}}{\underline\xi}\,\tau\,\sqrt{\log(2J)}.
\]
In particular, when $c=0$ (exact matching), the first term vanishes.
\end{theorem}

Theorem~\ref{thm:syn} shows that, in the absence of reference-domain–only latent factors ($\tilde\mu^F_i$), the bound on the bias of our proposed estimator vanishes as the number of time periods in the reference domain increases.
By contrast, the presence of target-domain–only latent factors ($\tilde\mu^Y_i$) does not prevent convergence.
We also see that the bound becomes wider as the number of donor units increases, due to the higher chance of overfitting, i.e., matching by chance.
A main ingredient in the proof is a maximal inequality to bound a linear combination of weighted error terms. This step also appears in \citep{abadie2010synthetic} in the standard synthetic control setting. However, in contrast to that work, we employ a sub-Gaussian approach (see Assumption~\ref{assm:factormodel}), leading to a clearer and more interpretable derivation without additional moment assumptions on auxiliary variables used in \citep{abadie2010synthetic}.
There are also challenges unique to the data fusion setting that do not arise in the standard synthetic control setting:
(i) we must handle target-domain–only and reference-domain–only latent factors; and
(ii) we must account for two sets of factor loadings for the shared latent factors $\mu_i$ across the two domains, which necessitates bounds for both sets. As seen in Theorem~\ref{thm:syn}, the upper bounds on both sets of loadings appear in the bias bound.

In the next section, we evaluate the performance of our proposed estimator in a simulation study and compare it with the estimators introduced in Section~\ref{sec:eqcon}.

\section{Simulation Studies}
\label{sec:sims}

In this section, we evaluate the performance of our proposed estimation methods. We consider a setting with 31 units ($J=30$) and observe $d_r = 3$ and $d_t = 3$ covariates in the reference and target domains, respectively. 
We also consider three latent confounders (i.e., $d_u = 3$) that affect the outcomes in both domains. We observe $T\in\{10,20,\ldots,100\}$ time periods in the reference domain and $S=5$ time periods in the target domain. Data are generated from the factor model in Display~\eqref{eq:factor}. 

\begin{figure}[t]
    \centering
    \subfigure[Evaluation of linear equi-confounding assumption on one DGP.]{%
        \includegraphics[width=0.45\textwidth]{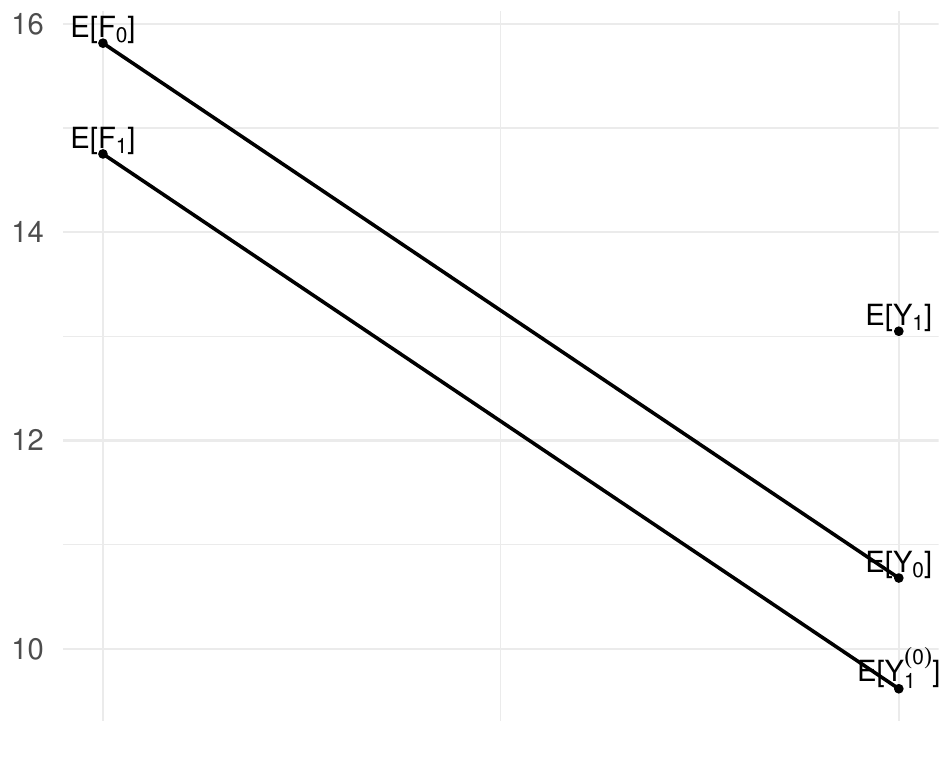}
        \label{fig:equi_linear}
    }
    \hfill
    \subfigure[Evaluation of logarithmic equi-confounding assumption on one DGP.]{%
        \includegraphics[width=0.45\textwidth]{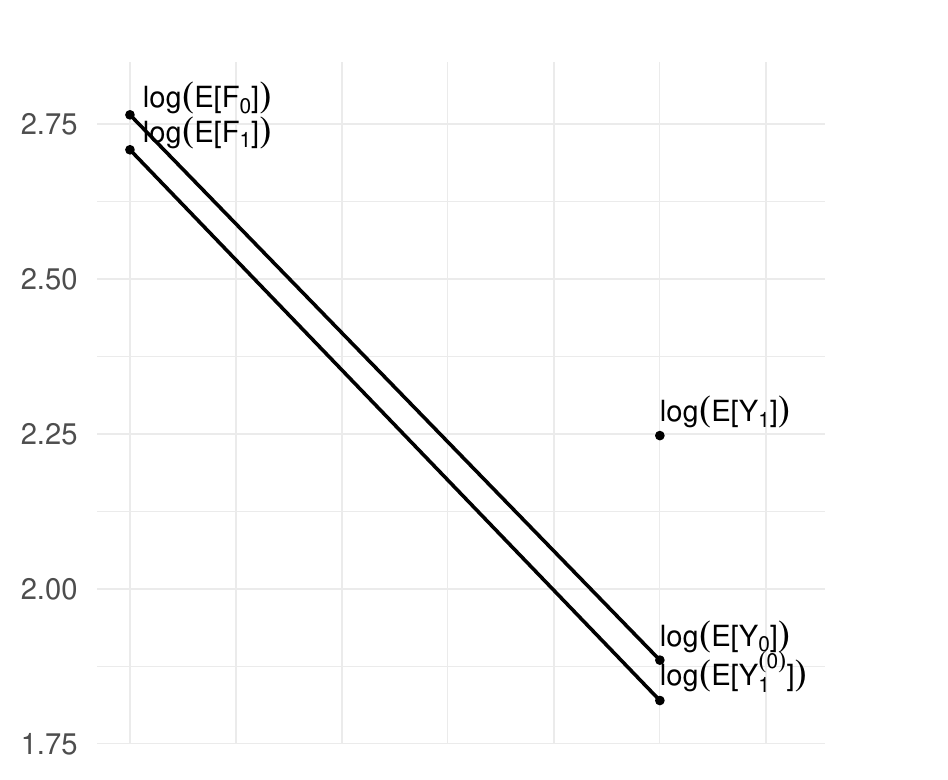}
        \label{fig:equi_log}
    }
    \caption{Visual demonstration of the equi-confounding assumptions.}
    \label{fig:plot_eq}
\end{figure} 

For all $i$, $s$, and $t$, the entries of $X_i$, $Z_i$, and $\mu_i$ are generated i.i.d.\ from $\operatorname{Unif}(0,1)$; the entries of $\phi_t$, $\theta_t$, $\varphi_s$, and $\vartheta_s$ are generated i.i.d.\ from $\operatorname{Unif}(0,10)$. The intercept $\rho_t$ is generated i.i.d.\ from $\operatorname{Unif}(0,20)$ and sorted in increasing order; the same procedure is used to generate $\varrho_s$ but with $\operatorname{Unif}(0,10)$. The noise terms $\epsilon_{i,t}$ and $\varepsilon_{i,s}$ are generated i.i.d.\ from the Gaussian distributions $\mathcal{N}(0,2)$ and $\mathcal{N}(0,0.5)$, respectively. Finally, the parameters $\alpha_{1,s}$ corresponding to the causal effect on the target unit are generated i.i.d.\ from $\operatorname{Unif}(2,5)$ and sorted in increasing order. 
Fixing the factors in the data-generating process (DGP) as above, we generate $M=300$ datasets from the DGP to evaluate the methods. To assess the effect of increasing $T$ from $T=T_1$ to $T=T_2>T_1$, the first $T_1$ factors in the latter are kept identical to those in the former. For an estimator $\hat{\psi}_T$ operating over $T$ reference periods, we estimate the bias using $M$ datasets as
\(
\frac{1}{M} \sum_{m=1}^{M} \hat{\psi}_{T}(D_m) - \psi_0
\),
where $D_m$ denotes the $m$-th dataset.

For the linear equi-confounding approach, let $F_0 = \frac{1}{J}\sum_{i=2}^{J+1} F_i$ and $Y_0 = \frac{1}{J}\sum_{i=2}^{J+1} Y_i$. By choosing the parameters $\varphi_s$ and $\phi_t$ to be $0$ for all $s$ and $t$, and requiring $\frac{1}{T} \sum_{t=1}^T\theta_t$ to be equal to $\frac{1}{S} \sum_{s=1}^S\vartheta_s$, Assumption~\ref{assm:equi} is satisfied: as shown in Figure~\ref{fig:equi_linear} (which is based on a realization of the DGP with $T=20$), the line connecting $\E[F_1]$ to $\E[Y_1^{(0)}]$ is parallel to the line connecting $\E[F_0]$ to $\E[Y_0]$. By choosing the parameters $\varphi_s$ and $\phi_t$ to be $0$ for all $s$ and $t$, and requiring $(\frac{1}{T} \sum_{t=1}^T\rho_t)/(\frac{1}{S} \sum_{s=1}^S\varrho_s)$ to be equal to $(\frac{1}{T} \sum_{t=1}^T\theta_t)/(\frac{1}{S} \sum_{s=1}^S\vartheta_s)$, Assumption~\ref{assm:logequi} is satisfied: as shown in Figure~\ref{fig:equi_log} (which is based on the realization of the DGP with $T=20$), the lines connecting $\log\{\E[F_1]\}$ to $\log\{\E[Y_1^{(0)}]\}$ and $\log\{\E[F_0]\}$ to $\log\{\E[Y_0]\}$ are parallel.


Note that Assumptions~\ref{assm:equi} and \ref{assm:logequi} need not hold for the factor model in Equation~\eqref{eq:factor} in general. To assess the robustness of the equi-confounding estimators, we evaluated the degree to which Assumptions \ref{assm:equi} and \ref{assm:logequi} are violated under the factor model specified in Assumption \ref{assm:factormodel} with the aforementioned choices of the distributions. We considered 1000 DGPs, holding $\mu_i$, $Z_i$, $X_i$ fixed while varying time-specific parameters $\rho_t$, $\phi_t$, $\theta_t$, $\varrho_s$, $\varphi_s$, $\vartheta_s$ and errors across DGPs. The empirical distribution of linear equi-confounding violations exhibits a mean of 0.871 (SD = 1.100), while logarithmic equi-confounding violations show substantially smaller deviations with a mean of 0.014 (SD = 0.017). Both distributions are centered near zero, indicating that the assumptions hold approximately on average. However, the logarithmic equi-confounding assumption demonstrates markedly tighter concentration around zero, suggesting it is substantially less sensitive to the particular realization of the DGP parameters.

Turning to the synthetic control method, note that the optimization problem in Equation~\eqref{optim:algo} has a quadratic objective with quadratic constraints. We solve it using the \texttt{gurobi} package \citep{gurobi} in \textsf{R}.
A key aspect of our methodology is the trade-off between the matching quality in the reference and target domains.
This balance is governed by the flexibility allowed via the parameters $\eta_Z$ and $\eta_X$, which control the relative importance of matching in the reference and target domains, respectively. A larger $\eta_Z$ allows better matching in the target domain at the expense of the reference domain, whereas a larger $\eta_X$ allows better alignment in the reference domain but may reduce the quality of matching in the target domain. This trade-off can be tuned based on domain-specific considerations regarding the association between covariates and outcomes.
In our implementation, we set both $\eta_Z$ and $\eta_X$ to $0.1$. We also conducted a sensitivity analysis to examine the dependence of the estimated causal effect on $\eta_Z$ and $\eta_X$.
As seen in the Figure~\ref{fig:sensitivity_sim}, deviations of $\eta_Z$ and $\eta_X$ from $0.1$--either increasing or decreasing--exhibit minimal impact on the estimated causal effects. Specifically, they lead to changes in the estimated causal effects that are bounded by $0.005$.

\begin{figure}[t]
    \centering
    \subfigure[Sensitivity analysis for the simulation study.]{%
        \includegraphics[width=0.40\textwidth, height=0.23\textheight]{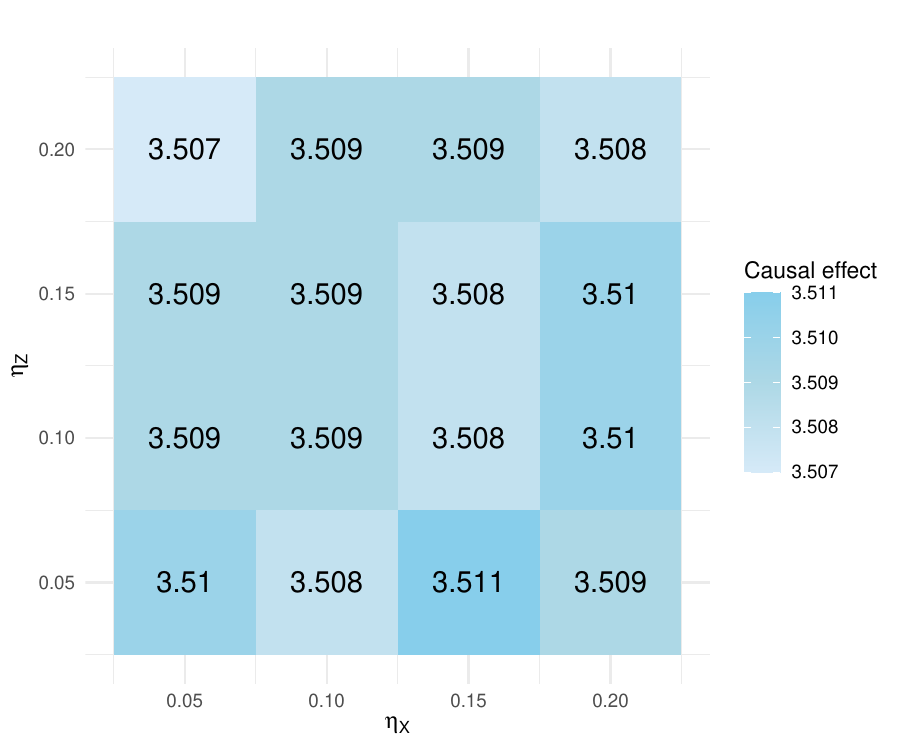}
        \label{fig:sensitivity_sim}
    }
    \hfill
    \subfigure[Sensitivity analysis for the real-data application.]{%
        \includegraphics[width=0.40\textwidth, height=0.23\textheight]{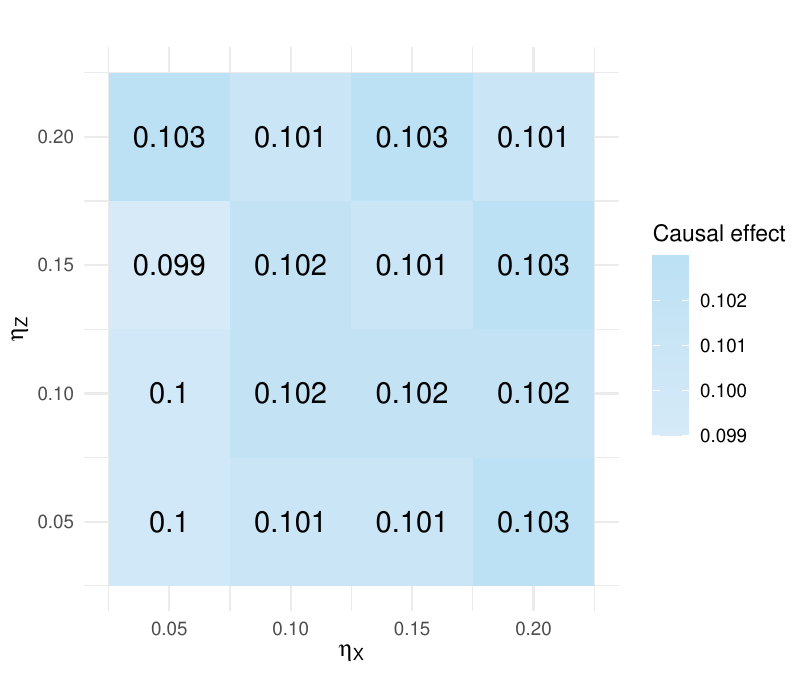}
        \label{fig:sensitivity_real_data}
    }
    \caption{Sensitivity analysis for the simulation study and the real-data application.}
    \label{fig:sensitivity_sim_real_data}
\end{figure}

\begin{figure}[t]
    \centering
    \subfigure[Units in the reference domain and the synthetic control target unit.]{%
        \includegraphics[width=0.45\textwidth, height=0.25\textheight]{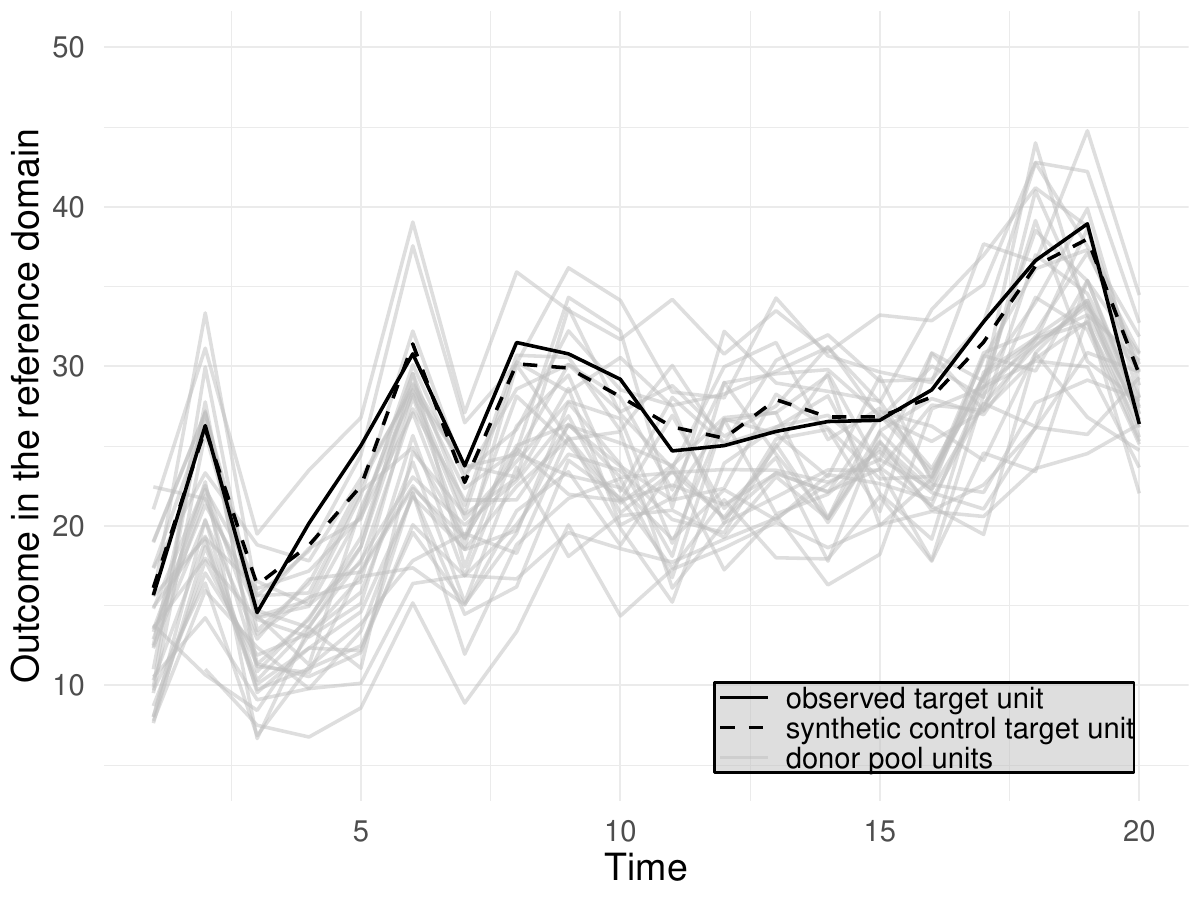}
        \label{fig:F_plot_sc}
    }
    \hfill
    \subfigure[Units in the target domain, the synthetic control target unit, and the counterfactual target unit.]{%
        \includegraphics[width=0.45\textwidth, height=0.25\textheight]{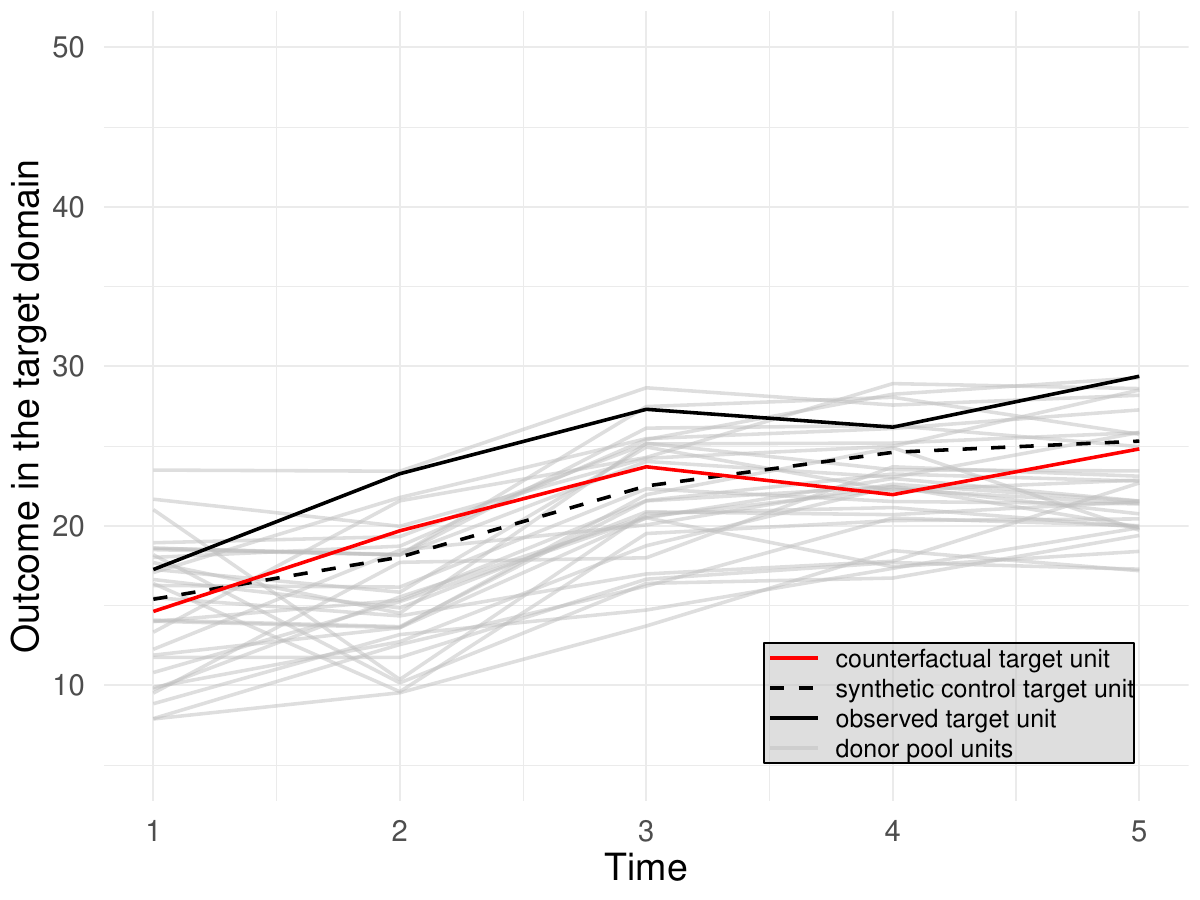}
        \label{fig:Y_plot_sc}
    }
    \caption{Factual, counterfactual, and estimated units under the synthetic control approach in the reference and target domains.}
    \label{fig:plot_sc}
\end{figure}

We draw a single dataset with $T = 20$. 
In Figure~\ref{fig:F_plot_sc}, the black solid line shows the trajectory of the observed target unit and the black dashed line shows the trajectory of 
$\sum_{i=2}^{J+1} w^*_{i} F_{i,t}$ for $t \in \{1,\ldots,T\}$. The gray solid lines represent the observed control units for $i \in \{2,\ldots,J+1\}$. We observe that the synthetic control target unit closely tracks the observed target unit in the reference domain. In Figure~\ref{fig:Y_plot_sc}, the black and gray solid lines represent the paths of the observed units for $i \in \{1,\ldots,J+1\}$ and $s \in \{1,\ldots,S\}$, and the solid red line represents the counterfactual target unit in the target domain. The dashed line shows the synthetic control target unit $\sum_{i=2}^{J+1} w^*_{i} Y_{i,s}$ using the same weights $w^*$ as above. Figure~\ref{fig:Y_plot_sc} demonstrates that the synthetic control data fusion method provides a good approximation to $Y_{1,s}^{(0)}$ using weights $w^*$ learned from the reference domain.

To demonstrate matching performance under our proposed method, we compute the NSEs for latent confounders, covariates, and outcomes in the reference domain.
Importantly, the NSE for $\bar F_1$ starts lower and then increases slightly before stabilizing. This is because, with shorter time periods, overfitting can occur and the algorithm may match noise; with longer time periods, overfitting is mitigated. Note that the NSE does not converge to zero, since the outcome variable contains noise at every time period; one should not expect an NSE smaller than the noise variance.
On the other hand, the NSE decreases for $Z_1$ because, as the time period length increases, better matching on $\bar F_1$ also improves matching on $Z_1$.
This indicates that the estimated matches become more accurate with more data points. Although the optimization does not directly target a good match for $\mu_1$, the optimized weights $w^*$ yield a good match between $\mu_1$ and $\sum_{i=2}^{J+1} w^*_{i} \mu_{i}$, as seen in Figure~\ref{fig:nse_mu_Z_X}. 

\begin{figure}[t]
    \centering
    \subfigure[NSEs of $\mu_i$, $Z_i$, $X_i$, and $\bar F_i$.]{%
        \includegraphics[width=0.31\textwidth, height=0.2\textheight]{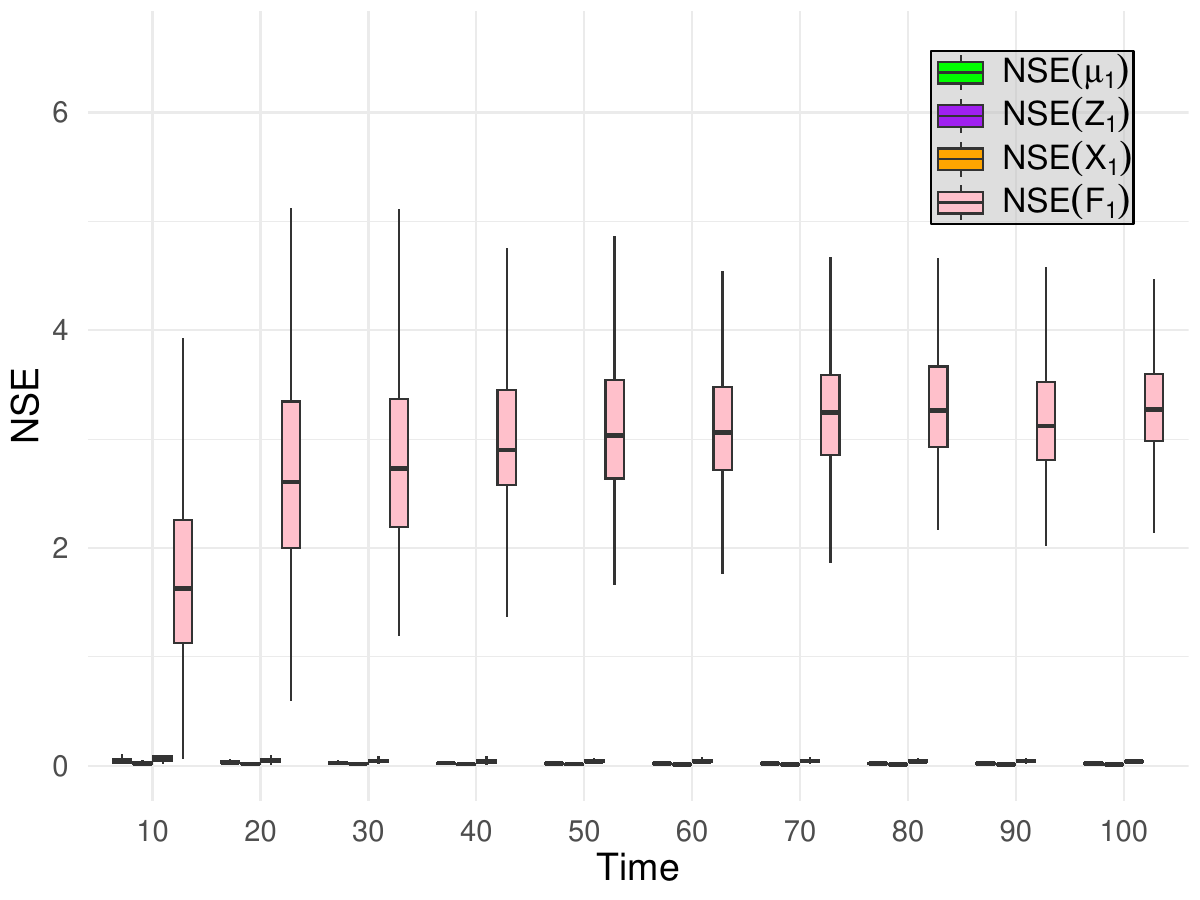}
        \label{fig:nse_mu_Z_X_F}
    }
    \hfill
    \subfigure[NSEs of $\mu_i$, $Z_i$, and $X_i$.]{%
        \includegraphics[width=0.31\textwidth, height=0.2\textheight]{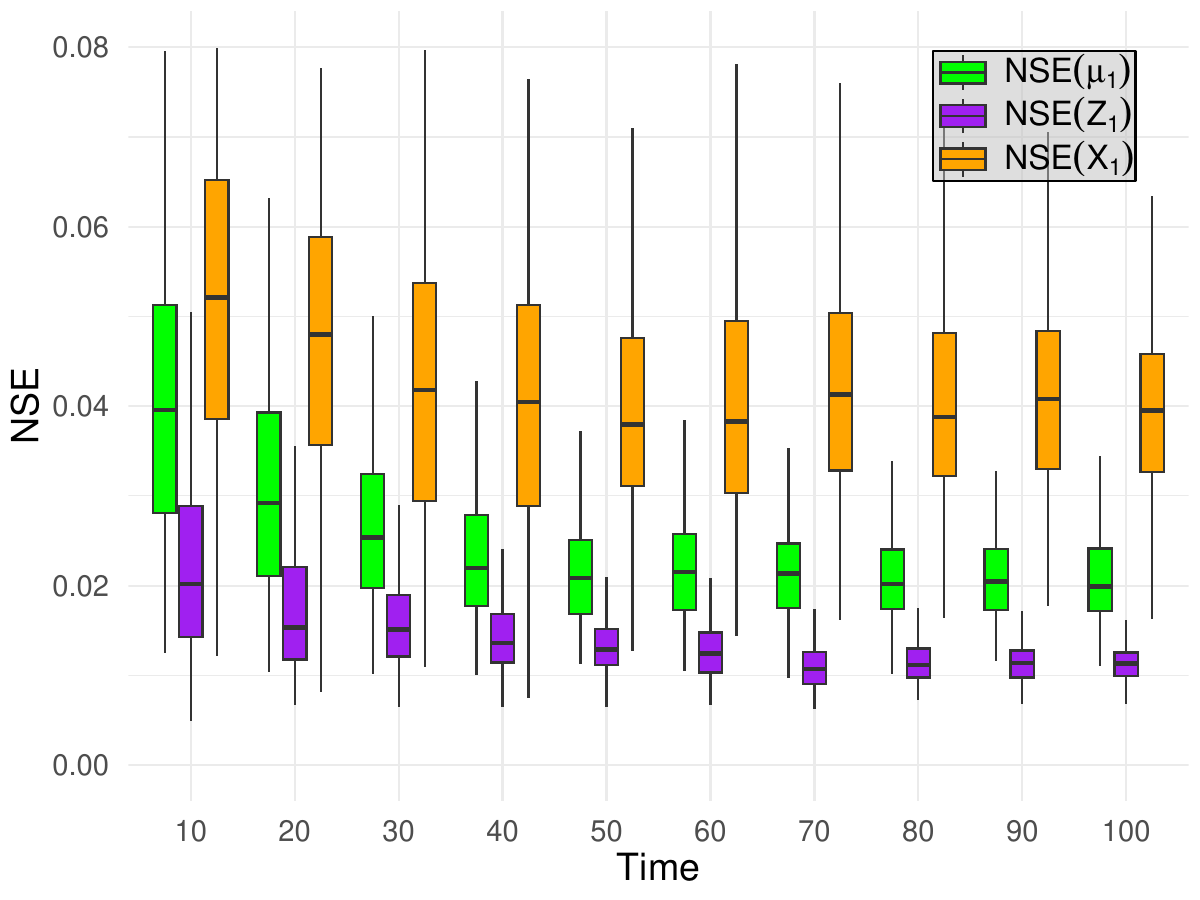}
        \label{fig:nse_mu_Z_X}
    }
    \hfill
    \subfigure[Differences of $\hat{\psi}^{{eq}1}$, $\hat{\psi}^{{eq}2}$, and $\hat{\psi}^{{sc}}$ from $\psi_0$.]{%
        \includegraphics[width=0.31\textwidth, height=0.2\textheight]{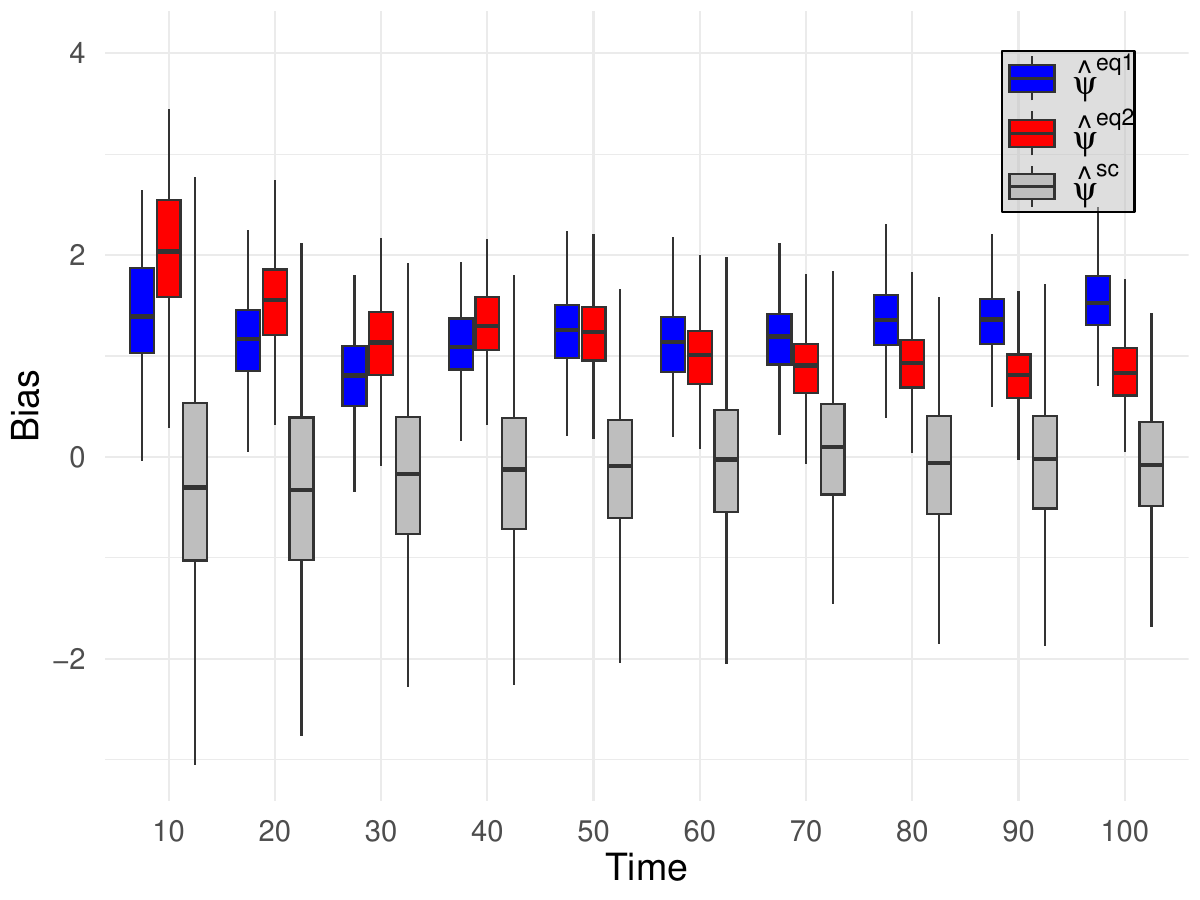}
        \label{fig:bias_sc_eq}
    }
    \caption{Bias and NSEs in synthetic control and equi-confounding methods (over 300 simulations) with a DGP satisfying Assumption \ref{assm:factormodel}.}
    \label{fig:NSEs_bias}
\end{figure}

\begin{figure}[t]
    \centering
    \subfigure[NSEs of $\mu_i$, $Z_i$, $X_i$, and $\bar F_i$.]{%
        \includegraphics[width=0.31\textwidth, height=0.2\textheight]{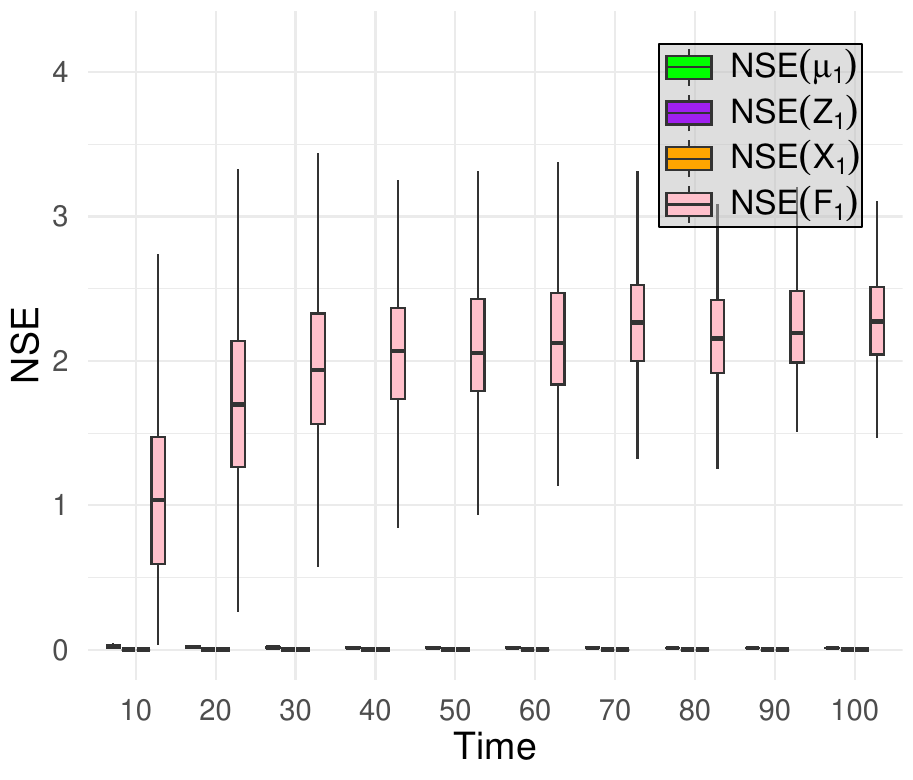}
        \label{fig:nse_mu_Z_X_F_300sim}
    }
    \hfill
    \subfigure[NSEs of $\mu_i$, $Z_i$, and $X_i$.]{%
        \includegraphics[width=0.31\textwidth, height=0.2\textheight]{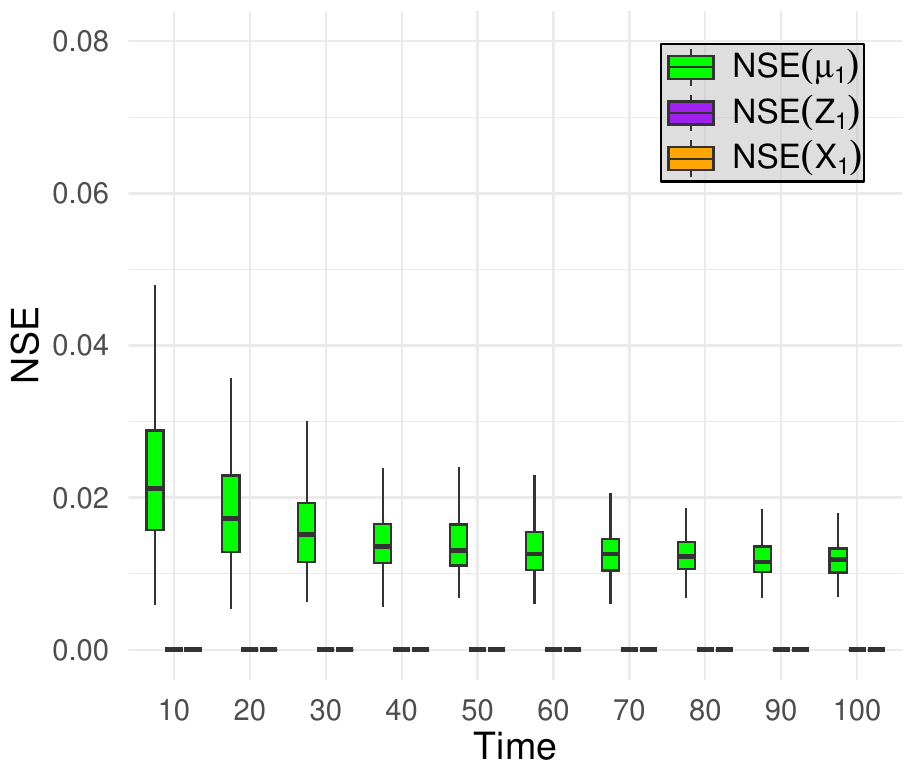}
        \label{fig:nse_mu_Z_X_300sim}
    }
    \hfill
    \subfigure[Differences of $\hat{\psi}^{{eq}1}$, $\hat{\psi}^{{eq}2}$, and $\hat{\psi}^{{sc}}$ from $\psi_0$.]{%
        \includegraphics[width=0.31\textwidth, height=0.2\textheight]{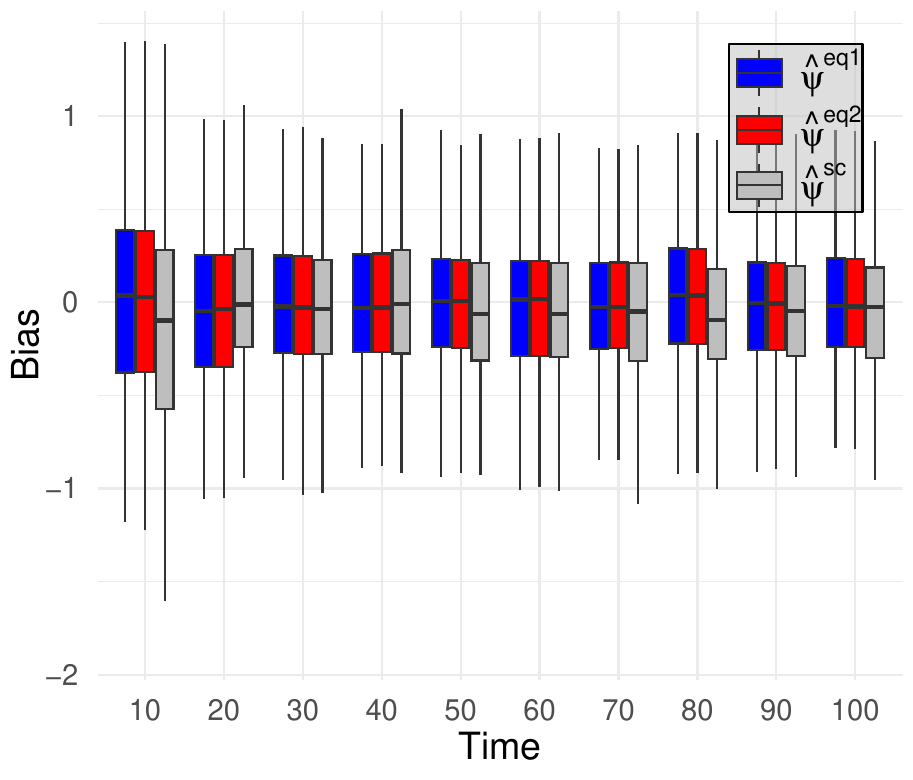}
        \label{fig:bias_sc_eq_300sim}
    }
    \caption{Bias and NSEs for synthetic control and equi-confounding methods (over 300 simulations) with a DGP satisfying the requirements of all frameworks.}
    \label{fig:NSEs_bias_300sim}
\end{figure}

Figure~\ref{fig:bias_sc_eq} provides a comparative visualization of the biases from the synthetic control data fusion and equi-confounding data fusion methods. We present the biases averaged over $300$ datasets for each value of $T$, along with the interquartile range. 
Overall, we observe superior performance from the synthetic control data fusion method.
The bias from the synthetic control data fusion method decreases as $T$ increases, eventually approaching zero, reaffirming Theorem~\ref{thm:syn} that the bias vanishes as the number of reference-domain periods grows. The logarithmic equi-confounding method also exhibits a decreasing bias trend. However, there is no clear decreasing pattern for the bias from the linear equi-confounding method. 
This is due to the fact that the main assumption of this method is not necessarily satisfied.
To examine this point, in Figure~\ref{fig:NSEs_bias_300sim}, we show performances of three proposed methods where the DGPs satisfy the requirements of all three methods. 

\begin{figure}[t]
    \centering
    \subfigure[Placebo test for the simulation study for the reference domain.]{%
        \includegraphics[width=0.45\textwidth, height=0.25\textheight]{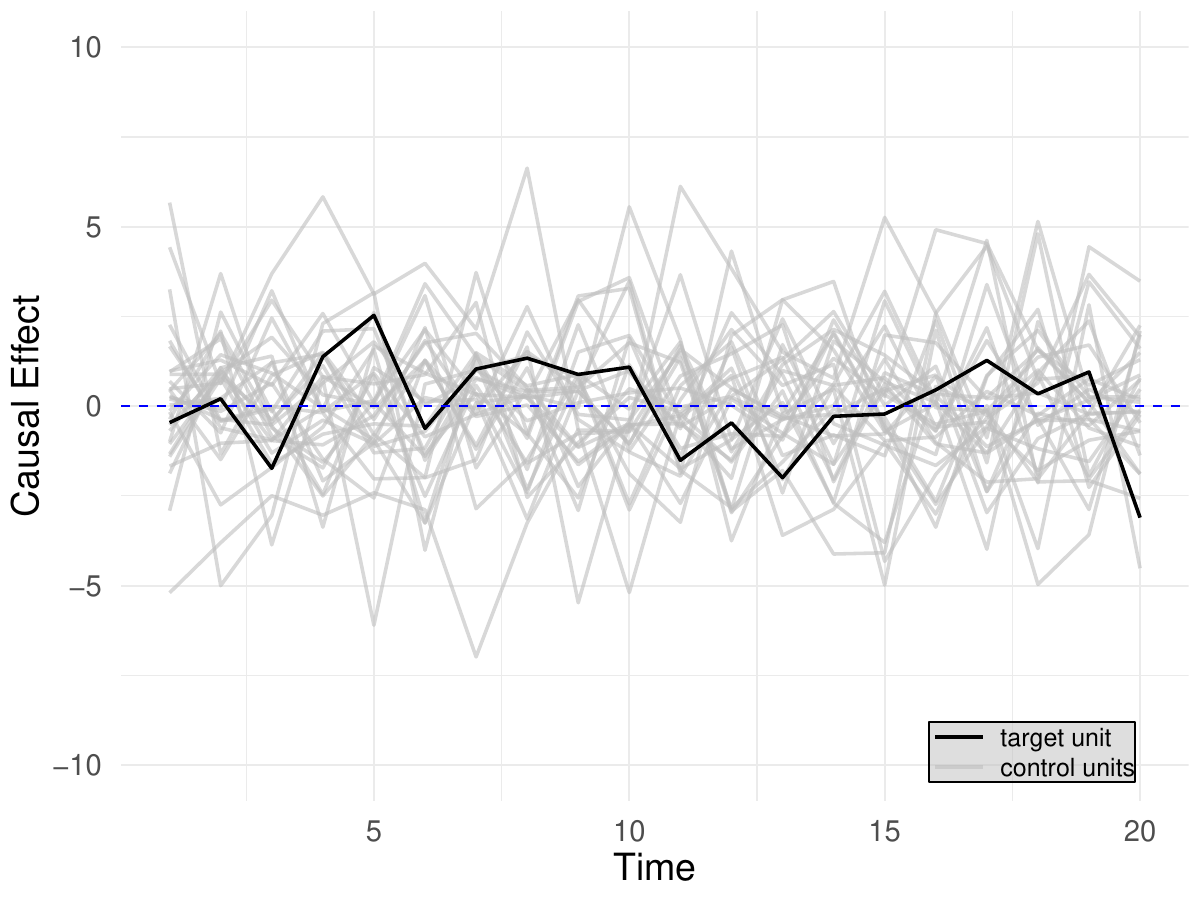}
        \label{fig:placeboF}
    }
    \hfill
    \subfigure[Placebo test for the simulation study for the target domain.]{%
        \includegraphics[width=0.45\textwidth, height=0.25\textheight]{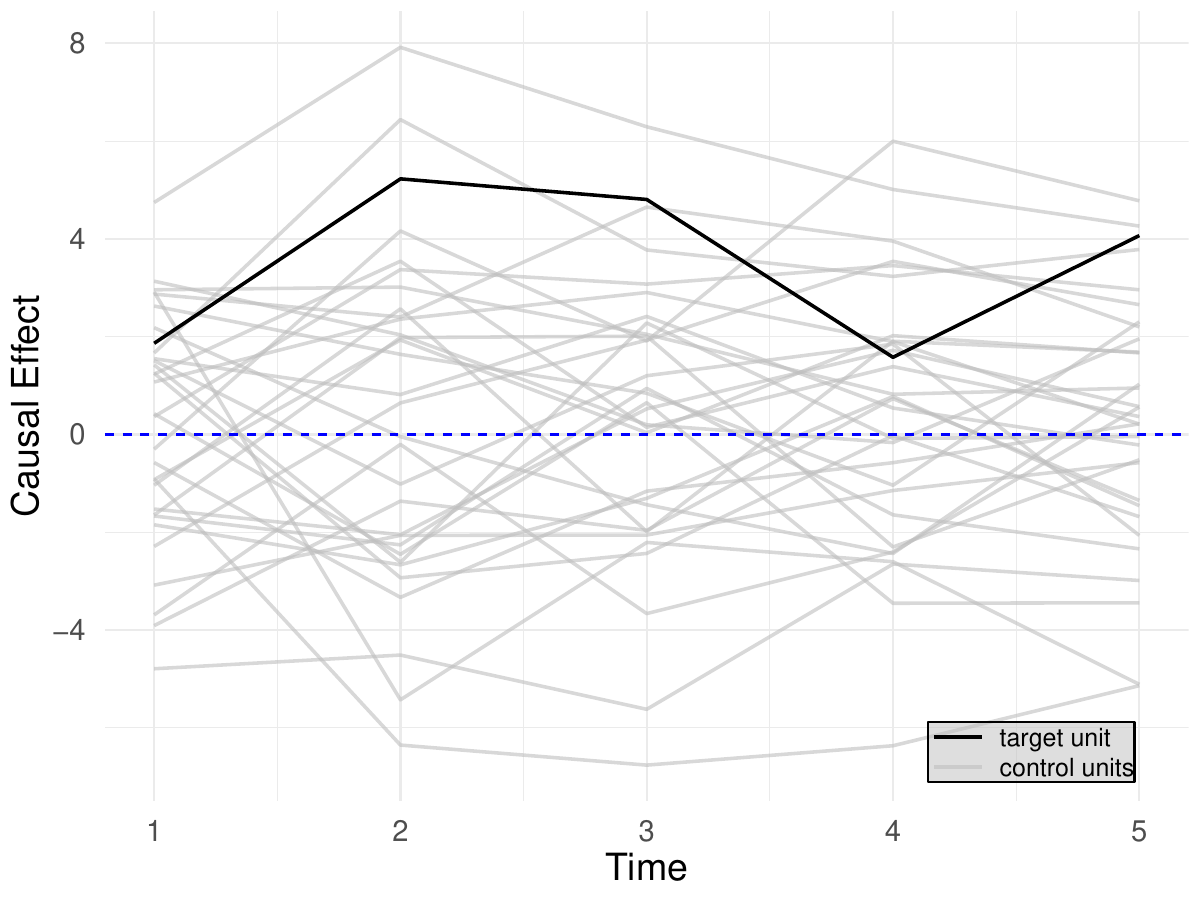}
        \label{fig:placeboT}
    }
    \subfigure[Placebo test for the COVID-19 vaccination rates application for Black sub-population (reference domain).]{%
        \includegraphics[width=0.45\textwidth, height=0.25\textheight]{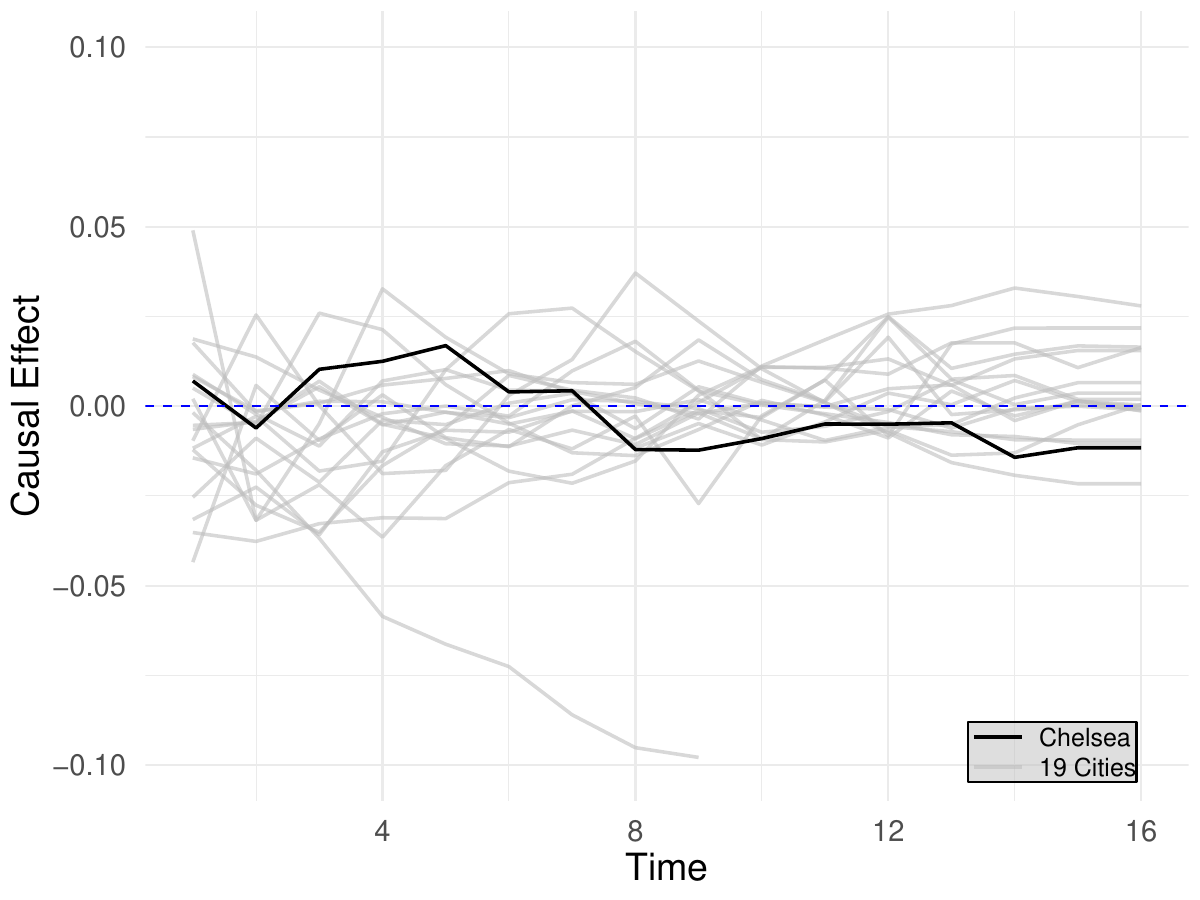}
        \label{fig:placeboF_data}
    }
    \hfill
    \subfigure[Placebo test for the COVID-19 vaccination rates application for Hispanic sub-population (target domain).]{%
        \includegraphics[width=0.45\textwidth, height=0.25\textheight]{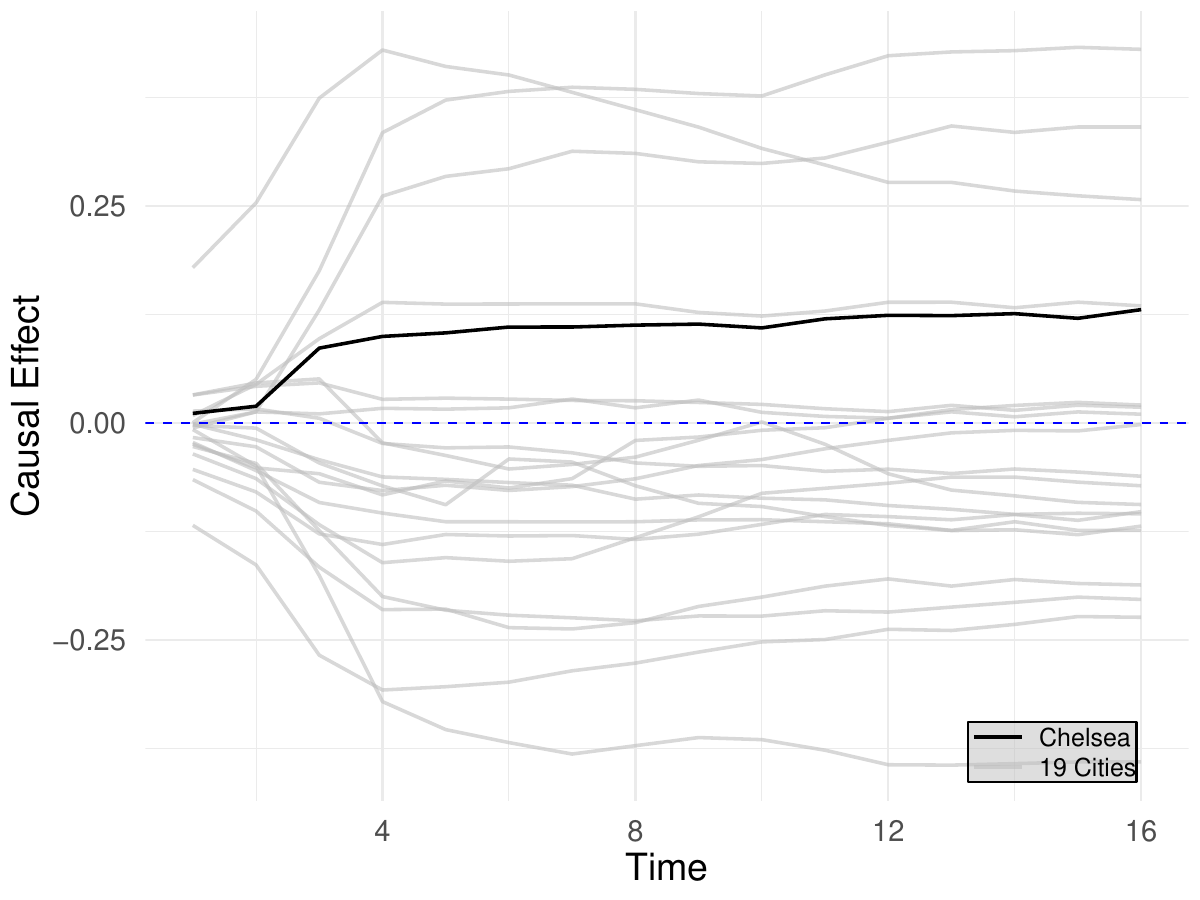}
        \label{fig:placeboT_data}
    }
    \caption{Placebo tests.}
    \label{fig:placebo}
\end{figure}

\begin{figure}[t]
    \centering
    \subfigure[Leave-one-out sensitivity test (reference domain).]{%
        \includegraphics[width=0.45\textwidth, height=0.25\textheight]{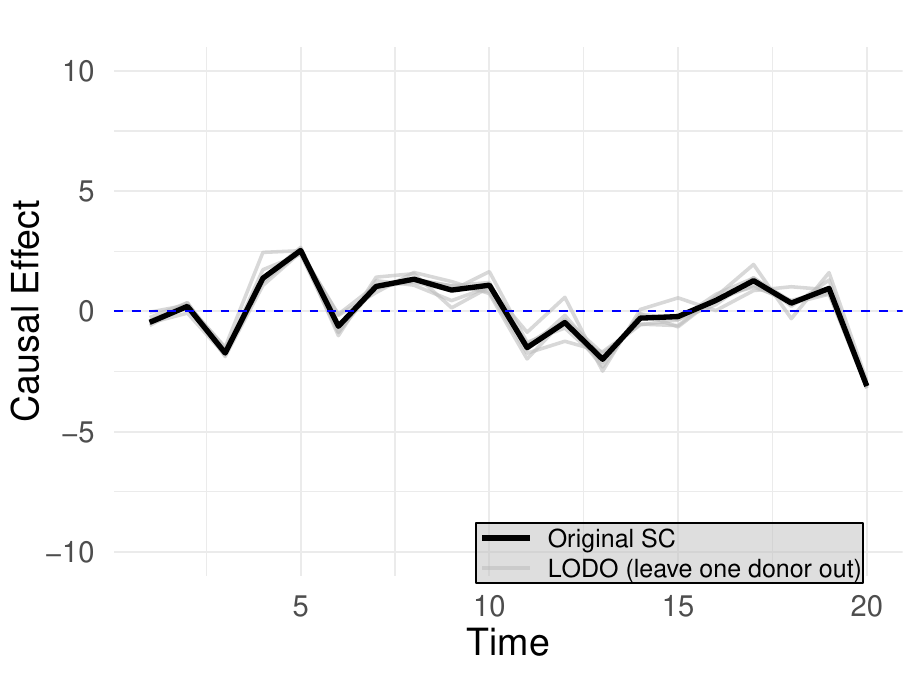}
        \label{fig:placeboF_LODO}
    }
    \hfill
    \subfigure[Leave-one-out sensitivity test (target domain).]{%
        \includegraphics[width=0.45\textwidth, height=0.25\textheight]{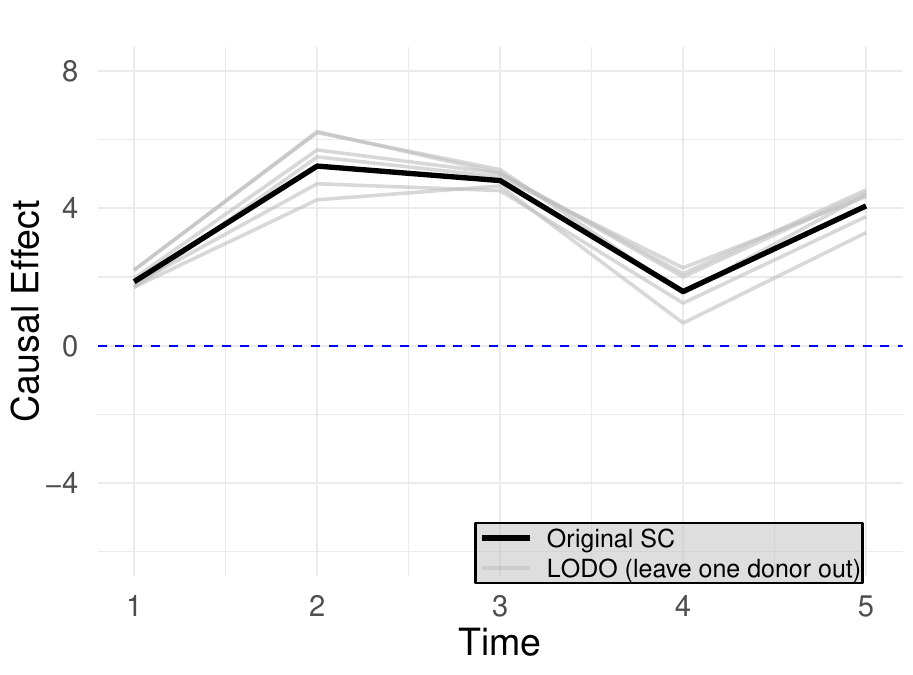}
        \label{fig:placeboY_LODO}
    }
    \subfigure[Leave-one-out sensitivity test for the COVID-19 vaccination rates application for Black sub-population (reference domain).]{%
        \includegraphics[width=0.45\textwidth, height=0.25\textheight]{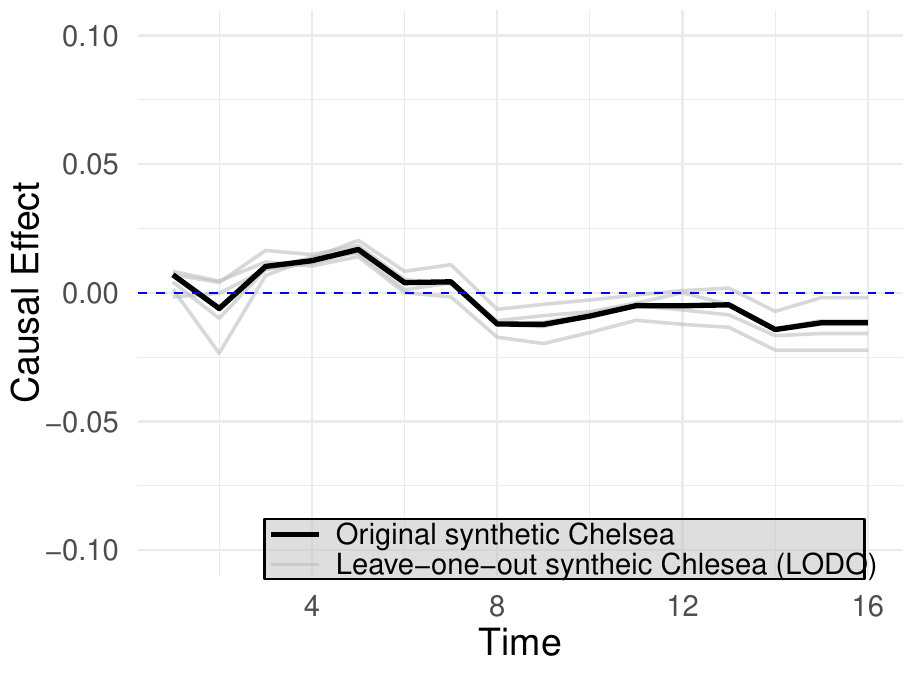}
        \label{fig:placebo_ref_LODO}
    }
    \hfill
    \subfigure[Leave-one-out sensitivity test for the COVID-19 vaccination rates application for Hispanic sub-population (target domain).]{%
        \includegraphics[width=0.45\textwidth, height=0.25\textheight]{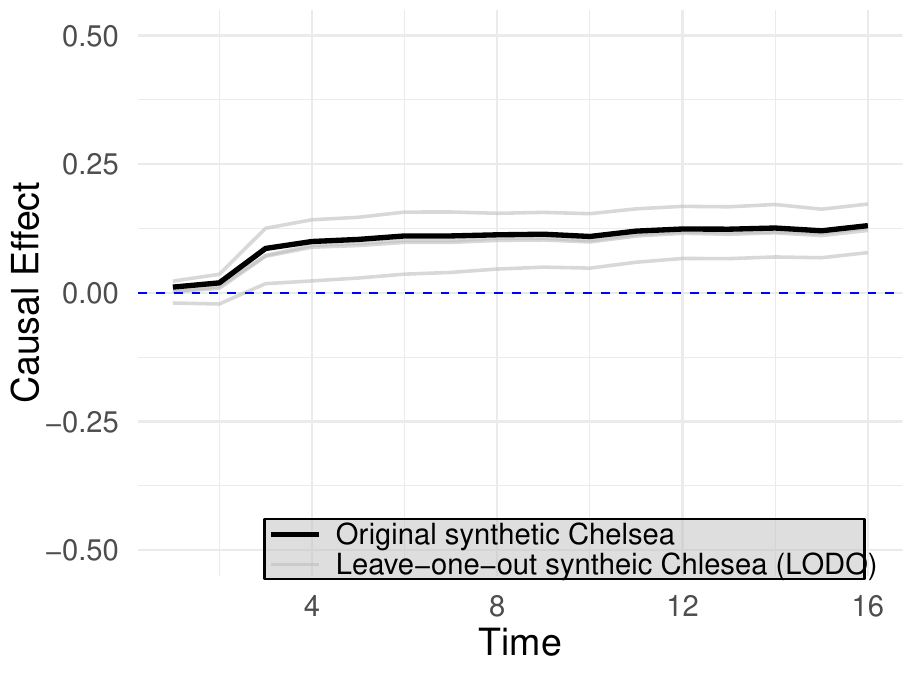}
        \label{fig:placebo_target_LODO}
    }
    \caption{Leave-one-out tests.}
    \label{fig:placebo_LODO}
\end{figure}

\subsection{Placebo Test}

To evaluate the ability of the synthetic control method to predict counterfactual outcomes, we implement a placebo test, an inferential technique originating in \citep{abadie2003economic}. Specifically, we apply the synthetic control method to each unit in the donor pool by generating synthetic controls using the same procedure employed for the target unit, shifting the target unit into the donor pool. We then compare the distribution of these placebo effects with the effect estimated for the target unit. 
By comparing the distribution of estimated placebo effects to the target-unit estimate, we can assess whether the observed effect is significantly larger than would be expected by chance.

The results of the placebo test are shown in Figure~\ref{fig:placebo}. The gray lines show the differences between each of the $30$ donor units and their synthetic controls, while the black line represents that for the target unit.
Figure~\ref{fig:placeboF} shows that the placebo effects are centered around zero, indicating negligible bias in the reference domain.
From Figure~\ref{fig:placeboT}, the estimated difference between the target unit and its synthetic control is substantially larger than that for the donors, demonstrating statistical significance of the estimated positive causal effect of the treatment. 

We also conducted a leave-one-out placebo test by iteratively removing each donor unit with a positive weight ($w > 0$) and re-estimating the synthetic control. Figures~\ref{fig:placeboF_LODO} and \ref{fig:placeboY_LODO} demonstrate that the leave-one-out estimates closely track the original synthetic treated unit. This confirms the robustness of our results across both the reference and target domains for the simulated data.

\section{Application to COVID-19 Vaccination Rate Analysis}
\label{sec:application}

During the COVID-19 pandemic, community organizations played a pivotal role in promoting health equity through localized efforts. 
In our empirical study, we evaluated the effect of the pro-vaccination initiatives of La Colaborativa, a major Hispanic community-based organization in Chelsea, Massachusetts, on the COVID-19 vaccination rate of the Hispanic sub-population of that city. Chelsea was among the 20 municipalities in Massachusetts most severely impacted by COVID-19. This disproportionate impact was evident in both COVID-19 case and death rates and in social determinants of health, particularly among Hispanic (which comprises $60\%$ of Chelsea’s population) and Black sub-populations. 

Despite the urgent need for vaccination, access barriers in the early rollout disproportionately affected non–English-speaking and low-income populations. When access first became available in Massachusetts, eligibility alone did not guarantee access—many residents faced challenges such as language barriers, lack of reliable transportation, and difficulties securing online appointments. These disparities created an urgent need for targeted community-based interventions to ensure equitable vaccine distribution. To address this challenge, La Colaborativa launched a series of pro-vaccination efforts in February/March 2021, in alignment with the Massachusetts Vaccine Equity Initiative. These initiatives included door-to-door outreach, provision of informational materials during free community food distribution, and Facebook Live information sessions, as well as a strategic partnership with the East Boston Neighborhood Health Center to increase accessibility for Chelsea’s Hispanic population. Later, a mobile vaccination clinic was offered by Massachusetts General Hospital.

Evaluating the impact of these interventions poses a significant challenge due to the timing of the available data. The Massachusetts Department of Public Health began recording COVID-19 vaccination rates in March 2021—coinciding with the start of La Colaborativa’s efforts. Hence, no pre-intervention data were collected. This lack of a pre-intervention period prevents the use of traditional causal inference approaches, which rely on pre-treatment trends to establish counterfactual comparisons. To address this limitation, we applied our proposed methodology to estimate the causal effects of La Colaborativa’s intervention.

\begin{figure}[t]
    \centering
    \subfigure[COVID-19 vaccination rates for Black sub-population (reference domain).]{%
        \includegraphics[width=0.45\textwidth, height=0.25\textheight]{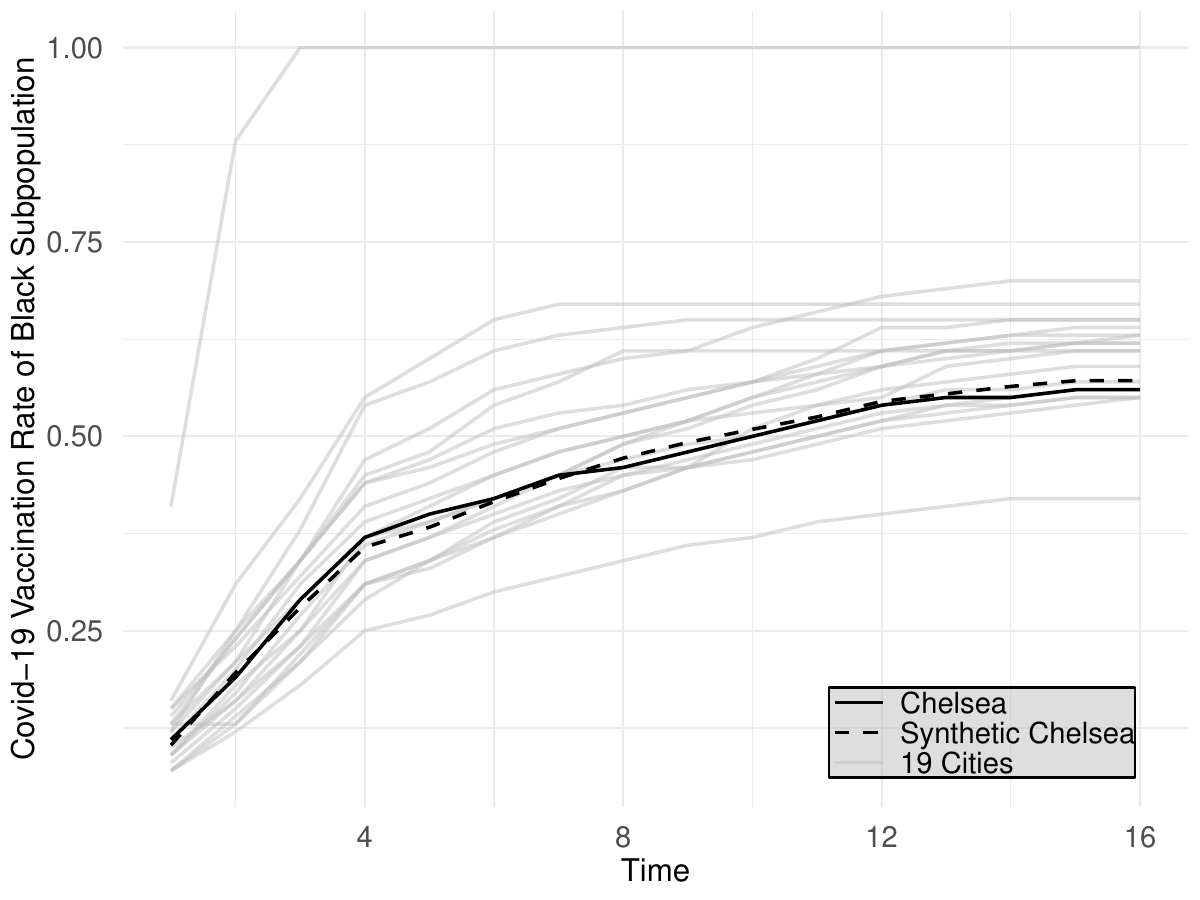}
        \label{fig:real_data_ref}
    }
    \hfill
    \subfigure[COVID-19 vaccination rates for Hispanic sub-population (target domain).]{%
        \includegraphics[width=0.45\textwidth, height=0.25\textheight]{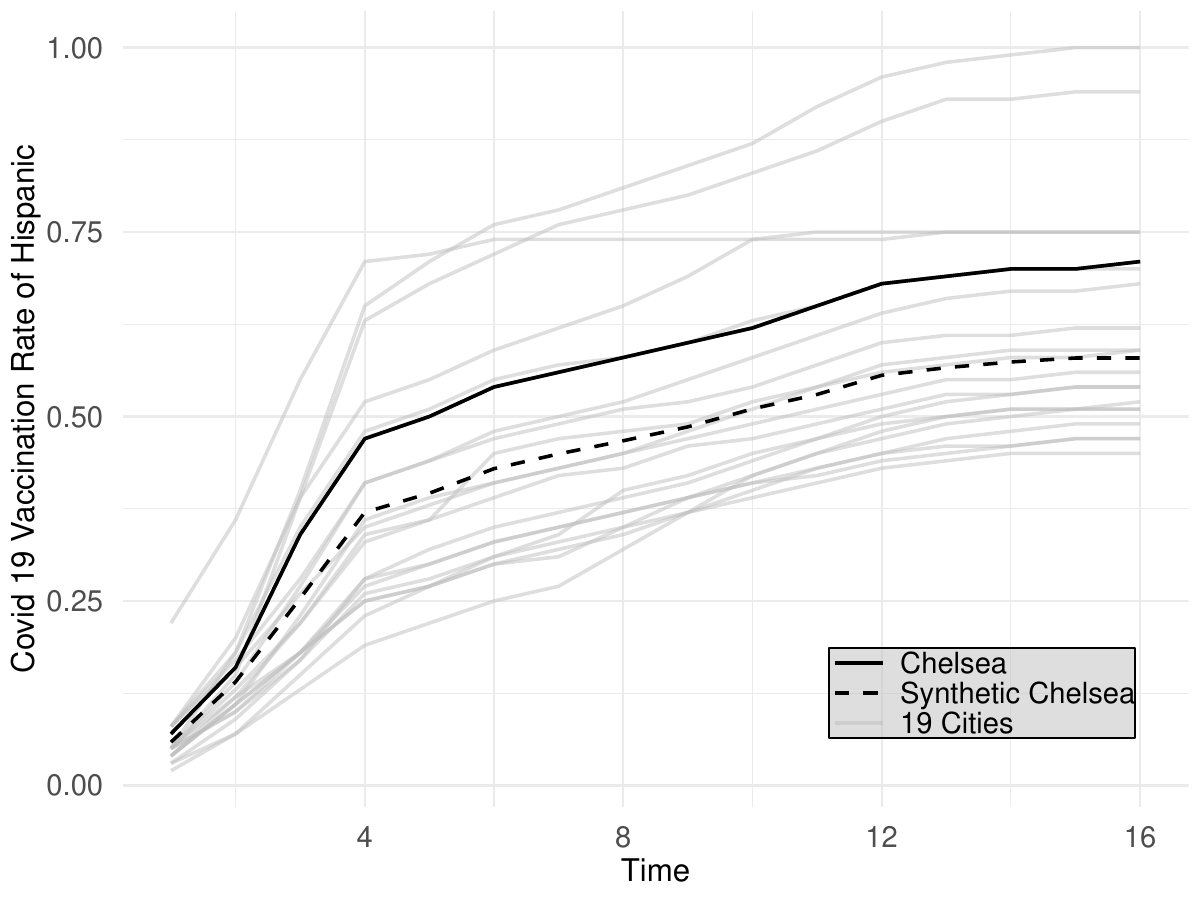}
        \label{fig:real_data_target}
    }
    \caption{COVID-19 vaccination rates for the observed units and synthetic Chelsea in reference and target domains.}
    \label{fig:real_data_refe_target}
\end{figure}

We considered the Hispanic sub-population in 20 cities in Massachusetts as our target domain, with monthly COVID-19 vaccination rate from March 2021 to July 2022 as the outcome variable, and with covariates median income, proportion of the sub-population, median age, and proportion of residents aged 65 or older of the Hispanic sub-population. 
The city of Chelsea, which is our target unit, is among these 20 cities. The remaining cities were chosen based on similarities in their covariates to Chelsea. Importantly, none of the other cities were served by community organizations with reach and staffing comparable to La Colaborativa. We considered the Black sub-population in the same 20 cities as our reference domain. The rationale for this choice is that both Black and Hispanic populations face systemic barriers to healthcare access, potentially comparable levels of vaccine hesitancy or exposure to misinformation on social media, and similar structural inequalities. These shared factors make Black vaccination trends a valuable reference for estimating causal effects in the Hispanic population.

\textbf{Results:} We used the value $0.1$ for both $\eta_Z$ and $\eta_X$. However, as in the simulation study, we performed sensitivity analysis to investigate the dependence of our estimated causal effect on the choice of $\eta_Z$ and $\eta_X$, and we observed a small impact of deviations from these hyperparameters. Figure~\ref{fig:sensitivity_real_data} presents the sensitivity analysis. As seen in this figure, deviations of $\eta_Z$ and $\eta_X$ from $0.1$--either increasing or decreasing--exhibit minimal impact on the estimated causal effects. Specifically, they lead to changes in the estimated causal effects that are bounded by $0.005$.
Table~\ref{tab:covariates} presents the covariates used in the estimation process. 
It shows that the proposed synthetic control method successfully identifies a combination of control units that closely approximates the covariates of Chelsea.
For example, the synthetic Chelsea is significantly closer than the average of all control cities across different covariates.
Although some discrepancies remain, such as in the ``Proportion - Hispanic''. This emphasizes the importance of using a method that allows approximate matching as in our Equations \eqref{eq:match1}--\eqref{eq:match3}.
Table~\ref{tab:predictors_unit_weights} displays the weights assigned to each control city in constructing synthetic Chelsea. The results indicate that the best matching for Chelsea is achieved by combining Revere, Holyoke, Everett, Lawrence, and Springfield. All other units in the donor pool were assigned zero weights.

\begin{table}
\caption{Covariate means for Chelsea, synthetic Chelsea, and the average of control units.}
\label{tab:covariates}
\centering
\small
\begin{tabular}{@{}lccc@{}}
\hline
Covariates & Chelsea & Synthetic Chelsea & Avg. of controls (n=19) \\
\hline
Median income -- Hispanic          & 0.522 & 0.388 & 0.367 \\
Median income -- Black             & 0.206 & 0.371 & 0.388 \\
Proportion -- Hispanic             & 0.784 & 0.479 & 0.294 \\
Proportion -- Black                & 0.362 & 0.350 & 0.310 \\
Median age -- Hispanic             & 0.730 & 0.526 & 0.396 \\
Median age -- Black                & 0.418 & 0.414 & 0.406 \\
Proportion aged 65+ -- Hispanic    & 0.637 & 0.348 & 0.265 \\
Proportion aged 65+ -- Black       & 0.230 & 0.309 & 0.235 \\
\hline
\end{tabular}
\end{table}

\begin{table}[t]
\caption{Unit weights in the synthetic Chelsea.}
\label{tab:predictors_unit_weights}
\centering
\small
\begin{tabular}{@{}lclc@{}}
\hline
Unit        & Weight & Unit        & Weight \\
\hline
Bellingham  & 0.000  & Marlborough & 0.000 \\
Chicopee    & 0.000  & Methuen     & 0.000 \\
Everett     & 0.228  & Milford     & 0.000 \\
Fitchburg   & 0.000  & New Bedford & 0.000 \\
Haverhill   & 0.000  & Peabody     & 0.000 \\
Holyoke     & 0.310  & Revere      & 0.330 \\
Lawrence    & 0.037  & Southbridge & 0.000 \\
Leominster  & 0.000  & Springfield & 0.094 \\
Lynn        & 0.000  & Waltham     & 0.000 \\
Worcester   & 0.000  &             &       \\
\hline
\end{tabular}
\end{table}

Figure~\ref{fig:real_data_refe_target} presents the COVID-19 vaccination rates for the observed units and synthetic Chelsea in the reference and target domains. In Figure~\ref{fig:real_data_ref}, synthetic Chelsea closely follows actual Chelsea in the reference domain, demonstrating that our algorithm has succeeded in obtaining a good match. This, along with the approximate matching in observed covariates, supports the validity of synthetic Chelsea as a counterfactual for actual Chelsea in the absence of the intervention. In Figure~\ref{fig:real_data_target}, the observed Chelsea Hispanic COVID-19 vaccination rate and its synthetic version begin to diverge after the intervention starts, with the gap widening over time. 
Specifically, Chelsea exhibits a steeper rise in the earlier months. The divergence between the black solid line (actual Chelsea) and the black dashed line (synthetic Chelsea) suggests that La Colaborativa’s efforts contributed to a greater increase in Hispanic vaccination rates than would have been expected without the intervention. 
We also provide the results of the placebo test in Figures~\ref{fig:placeboF_data} and \ref{fig:placeboT_data}, which confirm the statistical significance of our findings regarding the causal effect. We also conducted the leave-one-out placebo test, leaving out one of the five donor cities (Revere, Holyoke, Everett, Springfield and Lawrence) from the donor pool. In Figures~\ref{fig:placebo_ref_LODO} and \ref{fig:placebo_target_LODO}, the results show the five additional synthetic controls closely track the original synthetic Chelsea in both reference and target domain.

The estimated causal effect of the community organization intervention based on our proposed methodologies is as follows. Over the 16-month period from March 2021 to July 2022, the COVID-19 vaccination rate for the Hispanic sub-population is estimated to have increased by $13.6\%$ using the linear equi-confounding method, $13.2\%$ using the logarithmic equi-confounding method, and $10.1\%$ using the synthetic control data fusion method. All three methods indicate a similarly large positive causal effect, affirming the effectiveness of the community organization intervention.

\section{Conclusion}
\label{sec:conc}

We introduced two data fusion–based methods for identifying and estimating causal effects in panel data settings without pre-intervention data. The equi-confounding data fusion and synthetic control data fusion methods address challenges posed by sudden events that leave researchers without the pre-intervention period required by conventional panel-data approaches. We demonstrated the efficacy of these methods by deriving bounds on the bias and providing conditions under which these bounds converge to zero. We showed that the synthetic control data fusion method is more flexible, as it generally requires weaker assumptions, whereas the equi-confounding methods are simpler to implement and interpret. Our simulation results corroborate the theoretical findings and demonstrate rapidly decreasing biases. 

Our proposed causal panel-data fusion methods have high potential to be used in situations in which sudden shocks preclude the collection of pre-intervention data. As an illustrative application, we applied the proposed methods to estimate the causal effect of a community organization vaccination intervention in Chelsea, Massachusetts. Over the 16-month period from March 2021 to July 2022, the Hispanic COVID-19 vaccination rate in Chelsea is estimated to have increased by $13.6\%$ using the linear equi-confounding method, $13.2\%$ using the logarithmic equi-confounding method, and $10.1\%$ using the synthetic control data fusion method. All three methods indicate a similarly large positive causal effect, affirming the effectiveness of the community organization intervention.

\subsection*{Replication Data and \textsf{R} Package}
The replication code, data, and an \textsf{R} package are available in the following GitHub repositories:\\
Replication materials: 
\href{https://github.com/ZouYang31/Causal-Panel-Fusion}{https://github.com/ZouYang31/Causal-Panel-Fusion};\\
\textsf{R} package:
\href{https://github.com/ZouYang31/causalpanelfusion}{https://github.com/ZouYang31/causalpanelfusion}.


\subsection*{Acknowledgements}
\label{sec:acknow}
We are grateful to the Massachusetts Department of Public Health.
We extend special thanks to Boudu Bingay, Elizabeth T. Russo, and Joshua Norville for their support throughout this project.


\bibliographystyle{apalike}
\bibliography{references} 

\newpage

\appendices

\section{Proofs}

\section*{Proof of Theorem \ref{thm:equi}}
\label{sec:pf_eq}

Due to the consistency assumption, $\psi_0$ can be written as 
\[
\psi_0=\frac{1}{S}\sum_{s=1}^S\E[Y_{1,s}-Y_{1,s}^{(0)}]=\E[Y_1]-\E[Y_1^{(0)}].
\]
By Assumption \ref{assm:equi}, 
\begin{align*}
\E[Y_1^{(0)}]
&= 
\frac{1}{J}\sum_{i=2}^{J+1}\E[Y_i^{(0)}]
+\E[F_1]
-\frac{1}{J}\sum_{i=2}^{J+1}\E[F_i].
\end{align*}
Therefore,
\begin{align*}
\psi_0
&=\E[Y_1]-\E[Y_1^{(0)}]\\
&=\E[Y_1]-\frac{1}{J}\sum_{i=2}^{J+1}\E[Y_i^{(0)}]
-\E[F_1]
+\frac{1}{J}\sum_{i=2}^{J+1}\E[F_i].
\end{align*}
Finally, applying the consistency assumption one more time concludes the desired result.

\qed

\section*{Proof of Theorem \ref{thm:logequi}}
\label{sec:pf_logeq}

Due to the consistency assumption, $\psi_0$ can be written as 
\[
\psi_0=\frac{1}{S}\sum_{s=1}^S\E[Y_{1,s}-Y_{1,s}^{(0)}]=\E[Y_1]-\E[Y_1^{(0)}].
\]
By Assumption \ref{assm:logequi},
\begin{align*}
\E[Y_1^{(0)}]
&= 
\frac{ \E [F_1]}{\E[\sum_{i=2}^{J+1} F_i]} \E[\sum_{i=2}^{J+1} Y_i^{(0)}].
\end{align*}
Therefore, applying the consistency assumption one more time, we obtain
\begin{align*}
\psi_0
&= \E[Y_1]-\E[Y_1^{(0)}] \\
&=
\E[Y_1] - \frac{ \E [F_1]}{\E[\sum_{i=2}^{J+1} F_i]} \E[\sum_{i=2}^{J+1} Y_i].
\end{align*}

\qed
   
\section*{Proof of Theorem \ref{thm:logequibound}}
\label{sec:pf_logeqbound}
From Theorem \ref{thm:logequi}, the parameter of interest $\psi_0$ is identified as
\[
\psi_0 =\E[Y_1] - \frac{\E[F_1]}{\E[\sum_{i=2}^{J+1}F_i]} \E \left[\sum_{i=2}^{J+1} Y_i \right].
\]
For the estimator $\hat\psi^{eq2}$, we have
\begin{align*}
\E[\hat\psi^{eq2}]
&=\E[Y_1]-\E \left[\frac{F_1}{\sum_{i=2}^{J+1}F_i}\sum_{i=2}^{J+1}Y_i\right].
\end{align*}
Hence, using Assumption \ref{assm:bound}, the absolute bias of $\hat\psi^{eq2}$ will be bounded as
\begin{equation}
\label{eq:thm3-1}
\begin{aligned}
&\left|\E[\hat\psi^{eq2}]-\psi_0 \right|\\
&= 
\Bigg|
\E \left[\frac{F_1}{\frac{1}{J}\sum_{i=2}^{J+1}F_i} \right]\E \left[\frac{1}{J}\sum_{i=2}^{J+1}Y_i \right]
-\frac{\E[F_1]}{\E[\frac{1}{J}\sum_{i=2}^{J+1}F_i]} \E \left[ \frac{1}{J}\sum_{i=2}^{J+1} Y_i \right]\\
&\qquad+\E \left[\frac{F_1}{\frac{1}{J}\sum_{i=2}^{J+1}F_i} \frac{1}{J}\sum_{i=2}^{J+1}Y_i \right]
-\E \left[\frac{F_1}{\frac{1}{J}\sum_{i=2}^{J+1}F_i} \right]\E \left[\frac{1}{J}\sum_{i=2}^{J+1}Y_i \right]
\Bigg|\\
&= 
\Bigg|
\E \left[\frac{1}{J}\sum_{i=2}^{J+1}Y_i \right]\Bigg\{\E \left[\frac{F_1}{\frac{1}{J}\sum_{i=2}^{J+1}F_i} \right]
-\frac{\E[F_1]}{\E[\frac{1}{J}\sum_{i=2}^{J+1}F_i]}\Bigg\} 
+cov \left(\frac{F_1}{\frac{1}{J}\sum_{i=2}^{J+1}F_i}, \frac{1}{J}\sum_{i=2}^{J+1}Y_i \right)\Bigg|\\
&\le
L_y\left|  \E \left[\frac{F_1}{ \frac{1}{J}\sum_{i=2}^{J+1}F_i} \right] - \frac{\E[F_1]}{\E[ \frac{1}{J}\sum_{i=2}^{J+1}F_i]}  \right|+C(J).
\end{aligned}
\end{equation}
To simplify the notation, let $A:=F_1$ and $B:=\frac{1}{J}\sum_{i=2}^{J+1}F_i$. We aim to bound $\left|  \E \left[\frac{A}{B} \right] - \frac{\E[A]}{\E[B]}  \right|$. We note that
\begin{equation}
\label{eq:thm3-2}
\begin{aligned}
\left|  \E \left[\frac{A}{B} \right] - \frac{\E[A]}{\E[B]}  \right|
&=
\left|  \E \left[\frac{A}{B} \right]-\E[A]\E\left[\frac{1}{B}\right]+\E[A]\E\left[\frac{1}{B}
\right] - \frac{\E[A]}{\E[B]}  \right|\\
&=
\left| cov\left(A,\frac{1}{B}\right)+\E[A]\left\{\E\left[\frac{1}{B}
\right] - \frac{1}{\E[B]}\right\}  \right|\\
&\le
\left| cov\left(A,\frac{1}{B}\right)\right|+L_f\left|\E\left[\frac{1}{B}
\right] - \frac{1}{\E[B]}  \right|.
\end{aligned}
\end{equation}

\begin{lemma}
\label{lem:lipschitz_variance}
Let $g:\mathbb{R}\to\mathbb{R}$ be $L$-Lipschitz, i.e., $|g(x)-g(y)|\le L|x-y|$ for all $x,y$.
For any square-integrable random variable $Z$,
\[
\mathrm{var}\big(g(Z)\big)\;\le\; L^2\,\mathrm{var}(Z).
\]
\end{lemma}
\begin{proof}
Let $Z'$ be an independent copy of $Z$. Then
\[
\mathrm{var}\big(g(Z)\big)
=\frac{1}{2}E\!\left[\big(g(Z)-g(Z')\big)^2\right]
\le \frac{1}{2}E\!\left[L^2(Z-Z')^2\right]
=L^2\,\mathrm{var}(Z).
\]

\end{proof}

Using Cauchy–Schwarz inequality and Lemma \ref{lem:lipschitz_variance} with $g(x)=1/x$ on $[l_f,L_f]$, which is $1/l_f^2$-Lipschitz, we have
\begin{equation}
\label{eq:thm3-3}
 \left| cov\left(A,\frac{1}{B}\right)\right|
 \le\sqrt{var(A)var(\frac{1}{B})}
 \le\frac{1}{l_f^2}\sqrt{var(A)var(B)}.
\end{equation}

\begin{lemma}
\label{lem:second_derivative_gap}
Let $f$ be twice continuously differentiable on an interval containing the support of $B$.
Then
\[
\big|E[f(B)]-f(E[B])\big|
\;\le\;
\frac{1}{2}\,\sup_{x}\big|f''(x)\big|\;\mathrm{var}(B).
\]
\end{lemma}
\begin{proof}
Write $m=E[B]$. Taylor’s theorem with Lagrange remainder gives
$f(B)=f(m)+f'(m)(B-m)+\tfrac{1}{2}f''(\xi)(B-m)^2$ for some random $\xi$
between $B$ and $m$. Taking expectations and using $E[B-m]=0$ yields
$E[f(B)]-f(m)=\tfrac{1}{2}E\big[f''(\xi)(B-m)^2\big]$.
Taking absolute values and bounding $|f''(\xi)|$ by $\sup_x|f''(x)|$ proves the claim.

\end{proof}
\begin{sloppypar}
We apply Lemma \ref{lem:second_derivative_gap} with $f(x)=1/x$ on $[l_f,L_f]$. Here, $f''(x)=2/x^3$, hence $\sup_{x\in[l_f,L_f]}\big|f''(x)\big|=2/l_f^3$. This implies that
\end{sloppypar}
\begin{equation}
\label{eq:thm3-4}
    \Big|E\!\left[\frac{1}{B}\right]-\frac{1}{E[B]}\Big|
\;\le\;
\frac{1}{l_f^3}\,\mathrm{var}(B).
\end{equation}
Combining Equations \eqref{eq:thm3-2}-\eqref{eq:thm3-4}, we have
\begin{equation*}
\begin{aligned}
\left|  \E \left[\frac{A}{B} \right] - \frac{\E[A]}{\E[B]}  \right|
&\le
\left| cov\left(A,\frac{1}{B}\right)\right|+L_f\left|\E\left[\frac{1}{B}
\right] - \frac{1}{\E[B]}  \right|\\
&\le \frac{1}{l_f^2}\sqrt{var(A)var(B)}
+L_f\frac{1}{l_f^3}\,\mathrm{var}(B).
\end{aligned}
\end{equation*}

Each $F_i\in[l_f,L_f]$; hence $\mathrm{var}(F_i)\le \Delta_f^2/4$, giving $\mathrm{var}(A)\le \Delta_f^2/J$ and 
\[
\mathrm{var}(B)
=
\frac{1}{J^2}\sum_{i=2}^{J+1}\mathrm{var}(F_i)
+\frac{1}{J^2}\sum_{\substack{i,j=2\\ i\neq j}}^{J+1}\mathrm{cov}(F_i,F_j)
\;\le\;
\frac{1}{J^2}\sum_{i=2}^{J+1}\mathrm{var}(F_i) + \tau(J)
\;\le\;
\Delta_f^2/(4J)+\tau(J).
\]
 Therefore,
\begin{equation}
\label{eq:thm3-5}
\begin{aligned}
\left|  \E \left[\frac{A}{B} \right] - \frac{\E[A]}{\E[B]}  \right|
&\le \frac{1}{l_f^2}\sqrt{var(A)var(B)}
+L_f\frac{1}{l_f^3}\,\mathrm{var}(B)\\
&\le
\frac{1}{l_f^2}
\sqrt{\frac{\Delta_f^2}{4}\Big(\frac{\Delta_f^2}{4J}+\tau(J)\Big)}
\;+\;
\frac{L_f}{l_f^3}\Big(\frac{\Delta_f^2}{4J}+\tau(J)\Big)\\
&\le
\frac{\Delta_f^2}{4\,l_f^2}\,\frac{1}{\sqrt{J}}
\;+\;
\frac{\Delta_f}{2\,l_f^2}\sqrt{\tau(J)}
\;+\;
\frac{L_f}{l_f^3}\Big(\frac{\Delta_f^2}{4J}+\tau(J)\Big).
\end{aligned}
\end{equation}
Combining Equations \eqref{eq:thm3-1} and \eqref{eq:thm3-5} conclude the desired result.

\qed

\section*{Proof of Corollary \ref{cor2}}

The proof is the same as that of Theorem \ref{thm:logequibound}, except for our treatment of $\left| cov\left(A,\frac{1}{B}\right)\right|$ in Display \eqref{eq:thm3-2}, for which we argue as follows:
\begin{align*}
&\E[B]^2=B\E[B]-B(B-\E[B])+(B-\E[B])^2\\
&\Rightarrow\frac{1}{B}=\frac{1}{\E[B]}-\frac{(B-\E[B])}{\E[B]^2}+\frac{(B-\E[B])^2}{B\E[B]^2}.
\end{align*}
Therefore, we have
\begin{align*}
cov\left(A,\frac{1}{B}\right)
&=-\frac{cov(A,B)}{\E[B]^2}+cov\left(A,\frac{(B-\E[B])^2}{B\E[B]^2}\right)\\
&\le\frac{1}{l_f^2}|cov(A,B)|+\left|\E\left[\{A-\E[A]\}\frac{(B-\E[B])^2}{B\E[B]^2}\right]\right|\\
&\le\frac{1}{l_f^2}\tau_1(J)+\frac{\Delta_f}{l_f^3}\E[(B-\E[B])^2]\\&\le\frac{\tau_1(J)}{l_f^2}+\frac{\Delta_f}{l_f^3}var(B).
\end{align*}

\qed

\section*{Proof of Theorem \ref{thm:syn}}
\label{sec:pf_sc}

Consider the factor models in Assumption~\ref{assm:factormodel}:
\begin{align*}
    Y_{i,s} &= \varrho_s + \varphi_s^\top X_i + \vartheta_s^\top\mu_i
   + \tilde\vartheta_s^{\top} \tilde\mu_i^Y
    + \alpha_{i,s}A_{i,s} + \varepsilon_{i,s},\\
    F_{i,t} &= \rho_t + \phi_t^\top Z_i + \theta_t^\top\mu_i + \tilde\theta_t^{\top} \tilde\mu_i^F+ \epsilon_{i,t}.
\end{align*}
Under this model, the potential outcome under control satisfies
\[
Y_{i,s}^{(0)} = \varrho_s + \varphi_s^\top X_i + \vartheta_s^\top \mu_i
+ \tilde\vartheta_s^{\top} \tilde\mu_i^Y + \varepsilon_{i,s}.
\]

Define $\bar F_i := [F_{i,1}~\cdots~F_{i,T}]^\top$ and $\epsilon_i^R := [\epsilon_{i,1}~\cdots~\epsilon_{i,T}]^\top$.
Let $\phi^R$, $\theta^R$, and $\tilde\theta^R$ be the $T\times d_r$, $T\times d_u$, and
$T\times d^F_{\tilde u}$ matrices whose $t$th rows are $\phi_t^\top$, $\theta_t^\top$, and $\tilde\theta_t^\top$, respectively.

Fix $s\in\{1,\ldots,S\}$. 
The weighted average of the potential outcomes of the  donors is
\[
    \sum_{i=2}^{J+1} w_i Y_{i,s}=\sum_{i=2}^{J+1} w_i Y_{i,s}^{(0)} = \varrho_s  + \varphi_s^\top \sum_{i=2}^{J+1} w_i X_i + \vartheta_s^\top \sum_{i=2}^{J+1} w_i \mu_i 
    + \tilde\vartheta_s^{\top} \sum_{i=2}^{J+1} w_i \tilde\mu^Y_i
    + \sum_{i=2}^{J+1} w_i\varepsilon_{i,s}.
\]
Hence, we have
\begin{equation}
\label{eq:pf1_new}
Y_{1,s}^{(0)} - \sum_{i=2}^{J+1} w_i Y_{i,s}
= \varphi_s^\top\Big(X_1 - \sum_{i=2}^{J+1} w_i X_i\Big)
+ \vartheta_s^\top\Big(\mu_1 - \sum_{i=2}^{J+1} w_i \mu_i\Big)
+ \tilde\vartheta_s^\top\Big(\tilde\mu_1^Y - \sum_{i=2}^{J+1} w_i \tilde\mu_i^Y\Big)
+ \varepsilon_{1,s} - \sum_{i=2}^{J+1} w_i \varepsilon_{i,s}.
\end{equation}

Moreover,
\[
\bar F_1 - \sum_{i=2}^{J+1} w_i \bar F_i
= \phi^R\Big(Z_1 - \sum_{i=2}^{J+1} w_i Z_i\Big)
+ \theta^R\Big(\mu_1 - \sum_{i=2}^{J+1} w_i \mu_i\Big)
+ \tilde\theta^R\Big(\tilde\mu_1^F - \sum_{i=2}^{J+1} w_i \tilde\mu_i^F\Big)
+ \epsilon_1^R - \sum_{i=2}^{J+1} w_i \epsilon_i^R.
\]

Based on Assumption \ref{assm:factormodel}, the matrix $\theta^{R^\top}\theta^R$ is invertible.
Multiplying by $\vartheta_s^\top ({\theta^{R}}^\top \theta^R)^{-1}{\theta^{R}}^\top$ yields
\begin{align*}
   \vartheta_s^\top ({\theta^{R}}^\top \theta^R)^{-1}{\theta^{R}}^\top\{\bar F_1 - \sum_{i=2}^{J+1} w_i \bar F_i \} 
   =&~ \vartheta_s^\top ({\theta^{R}}^\top \theta^R)^{-1}{\theta^{R}}^\top \phi^R(Z_1 -\sum_{i=2}^{J+1} w_i Z_i) \\
    &+ \vartheta_s^\top ({\theta^{R}}^\top \theta^R)^{-1}{\theta^{R}}^\top\theta^R (\mu_1 - \sum_{i=2}^{J+1} w_i \mu_i) \\
&+\vartheta_s^\top ({\theta^{R}}^\top \theta^R)^{-1}{\theta^{R}}^\top\tilde\theta^R (\tilde\mu^F_1 - \sum_{i=2}^{J+1} w_i \tilde\mu^F_i) \\   
    &+\vartheta_s^\top ({\theta^{R}}^\top \theta^R)^{-1}{\theta^{R}}^\top\{ \epsilon_1^R - \sum_{i=2}^{J+1} w_i \epsilon_i^R\},
\end{align*}
which implies that
\begin{equation}
\label{eq:pf2}
\begin{aligned}
\vartheta_s^\top  (\mu_1 - \sum_{i=2}^{J+1} w_i \mu_i)
=&~\vartheta_s^\top ({\theta^{R}}^\top \theta^R)^{-1}{\theta^{R}}^\top\{\bar F_1 - \sum_{i=2}^{J+1} w_i \bar F_i \}\\ 
&- \vartheta_s^\top ({\theta^{R}}^\top \theta^R)^{-1}{\theta^{R}}^\top \phi^R(Z_1 -\sum_{i=2}^{J+1} w_i Z_i) \\
&-\vartheta_s^\top ({\theta^{R}}^\top \theta^R)^{-1}{\theta^{R}}^\top\tilde\theta^R (\tilde\mu^F_1 - \sum_{i=2}^{J+1} w_i \tilde\mu^F_i) \\ 
&-\vartheta_s^\top ({\theta^{R}}^\top \theta^R)^{-1}{\theta^{R}}^\top\{ \epsilon_1^R - \sum_{i=2}^{J+1} w_i \epsilon_i^R\}.
\end{aligned}
\end{equation}

Substituting into \eqref{eq:pf1_new} yields
\begingroup
\allowdisplaybreaks
\begin{equation}
\label{eq:bias_new}
\begin{aligned}
Y_{1,s}^{(0)} - \sum_{i=2}^{J+1} w_i Y_{i,s}
=&~ \varphi_s^\top\Big(X_1 - \sum_{i=2}^{J+1} w_i X_i\Big)\\
&+ \vartheta_s^\top(\theta^{R^\top}\theta^R)^{-1}\theta^{R^\top}
\Big(\bar F_1 - \sum_{i=2}^{J+1} w_i \bar F_i\Big)\\
&-\vartheta_s^\top(\theta^{R^\top}\theta^R)^{-1}\theta^{R^\top}\phi^R
\Big(Z_1 - \sum_{i=2}^{J+1} w_i Z_i\Big)\\
&-\vartheta_s^\top(\theta^{R^\top}\theta^R)^{-1}\theta^{R^\top}\tilde\theta^R
\Big(\tilde\mu_1^F - \sum_{i=2}^{J+1} w_i \tilde\mu_i^F\Big)\\
&-\vartheta_s^\top(\theta^{R^\top}\theta^R)^{-1}\theta^{R^\top}
\Big(\epsilon_1^R - \sum_{i=2}^{J+1} w_i \epsilon_i^R\Big)\\
&+ \tilde\vartheta_s^\top\Big(\tilde\mu_1^Y - \sum_{i=2}^{J+1} w_i \tilde\mu_i^Y\Big)
+ \varepsilon_{1,s} - \sum_{i=2}^{J+1} w_i \varepsilon_{i,s}.
\end{aligned}
\end{equation}
\endgroup

Define the following terms:
\begin{align*}
R_{1s} &:= \vartheta_s^\top(\theta^{R^\top}\theta^R)^{-1}\theta^{R^\top}\sum_{i=2}^{J+1} w_i \epsilon_i^R,\\
R_{2s} &:= -\vartheta_s^\top(\theta^{R^\top}\theta^R)^{-1}\theta^{R^\top}\tilde\theta^R
\Big(\tilde\mu_1^F - \sum_{i=2}^{J+1} w_i \tilde\mu_i^F\Big),\\
R_{3s} &:= -\vartheta_s^\top(\theta^{R^\top}\theta^R)^{-1}\theta^{R^\top}\epsilon_1^R,\\
R_{4s} &:= \varepsilon_{1,s} - \sum_{i=2}^{J+1} w_i \varepsilon_{i,s},\\
R_{5s} &:= \tilde\vartheta_s^\top\Big(\tilde\mu_1^Y - \sum_{i=2}^{J+1} w_i \tilde\mu_i^Y\Big),\\
R_{6s} &:= \varphi_s^\top\Big(X_1 - \sum_{i=2}^{J+1} w_i X_i\Big),\\
R_{7s} &:= \vartheta_s^\top(\theta^{R^\top}\theta^R)^{-1}\theta^{R^\top}
\Big(\bar F_1 - \sum_{i=2}^{J+1} w_i \bar F_i\Big),\\
R_{8s} &:= \vartheta_s^\top(\theta^{R^\top}\theta^R)^{-1}\theta^{R^\top}\phi^R
\Big(Z_1 - \sum_{i=2}^{J+1} w_i Z_i\Big).
\end{align*}
Then \eqref{eq:bias_new} can be written as
\[
Y_{1,s}^{(0)} - \sum_{i=2}^{J+1} w_i Y_{i,s}
= R_{6s} + R_{7s} - R_{8s} + R_{1s} + R_{2s} + R_{3s} + R_{4s} + R_{5s}.
\]

By Assumption \ref{assm:factormodel}, 
$\mathbb E[R_{3s}]=0$. Moreover, since
$\{\varepsilon_{i,s}\}_{i,s}$ and $\{\tilde\mu_i^Y\}_i$ are independent of $\mathcal G$ and have mean
zero, and $w$ is $\mathcal G$-measurable, we have $\mathbb E[R_{4s}]=\mathbb E[R_{5s}]=0$.
Therefore, by the triangle inequality,
\begin{equation}
\label{eq:triangle_new}
\left|\mathbb E\left[Y_{1,s}^{(0)} - \sum_{i=2}^{J+1} w_i Y_{i,s}\right]\right|
\le
|\mathbb E[R_{6s}]|+|\mathbb E[R_{7s}]|+|\mathbb E[R_{8s}]|
+|\mathbb E[R_{1s}]|+|\mathbb E[R_{2s}]|.
\end{equation}

\paragraph{Bounding $|\mathbb E[R_{6s}]|,|\mathbb E[R_{7s}]|,|\mathbb E[R_{8s}]|$ using conditions
\eqref{eq:match1}--\eqref{eq:match3}
.}
First, by Cauchy--Schwarz and $\|\varphi_s\|_2\le \bar\varphi$,
\[
|\mathbb E[R_{6s}]|
\le \|\varphi_s\|_2 \left\|X_1 - \sum_{i=2}^{J+1} w_i X_i\right\|_2
\le \bar\varphi \sqrt{d_t}\,c,
\]
where the last step used condition \eqref{eq:match3}

Let $b_s := \theta^R(\theta^{R^\top}\theta^R)^{-1}\vartheta_s\in\mathbb R^T$. Then
$\vartheta_s^\top(\theta^{R^\top}\theta^R)^{-1}\theta^{R^\top}=b_s^\top$ and
\[
|\mathbb E[R_{7s}]|
\le \|b_s\|_2\left\|\bar F_1 - \sum_{i=2}^{J+1} w_i \bar F_i\right\|_2.
\]
As in the standard argument,
\[
\|b_s\|_2
\le \|\theta^R\|_{\mathrm{op}}\;\|(\theta^{R^\top}\theta^R)^{-1}\|_{\mathrm{op}}\;\|\vartheta_s\|_2
\le \frac{\sqrt{T}\,\bar\theta}{T\underline\xi}\,\bar\vartheta
= \frac{\bar\vartheta\,\bar\theta}{\underline\xi}\,\frac{1}{\sqrt{T}}.
\]
Combining with condition \eqref{eq:match1}
i.e., \ $\|\bar F_1-\sum_i w_i\bar F_i\|_2\le \sqrt{T}\,c$, gives
\[
|\mathbb E[R_{7s}]| \le \frac{\bar\vartheta\,\bar\theta}{\underline\xi}\,c.
\]

Similarly,
\[
|\mathbb E[R_{8s}]|
\le \|b_s^\top \phi^R\|_2 \left\|Z_1 - \sum_{i=2}^{J+1} w_i Z_i\right\|_2
\le \|b_s\|_2\,\|\phi^R\|_{\mathrm{op}}\,\sqrt{d_r}\,c,
\]
using condition \eqref{eq:match2}.
Finally, $\|\phi^R\|_{\mathrm{op}}\le \sqrt{\sum_{t=1}^T\|\phi_t\|_2^2}\le \sqrt{T}\,\bar\phi$, hence
\[
|\mathbb E[R_{8s}]|
\le \left(\frac{\bar\vartheta\,\bar\theta}{\underline\xi}\,\frac{1}{\sqrt{T}}\right)\left(\sqrt{T}\,\bar\phi\right)\sqrt{d_r}\,c
=
\frac{\bar\vartheta\,\bar\theta\,\bar\phi}{\underline\xi}\,\sqrt{d_r}\,c.
\]

\paragraph{Bounding $|\mathbb E[R_{1s}]|$ and $|\mathbb E[R_{2s}]|$.}

We re-write $R_{1s}$ as
\begin{align*}
    R_{1s} &= \vartheta_s^\top ({\theta^{R}}^\top \theta^R)^{-1}{\theta^{R}}^\top \sum_{i=2}^{J+1} w_i \epsilon_i^R\\
    &= \sum_{i=2}^{J+1} {w_i} \vartheta_s^\top ({\theta^{R}}^\top \theta^R)^{-1} {\theta^{R}}^\top  \epsilon_i^R\\
    &= \sum_{i=2}^{J+1} {w_i} b_s^\top \epsilon_i^R,
\end{align*}
where $b_s={\theta^{R}}({\theta^{R}}^\top \theta^R)^{-1}\vartheta_s$. Notice that $\epsilon_i^R$ has independent mean-zero sub-Gaussian entries with variance proxy $\bar\sigma^2$.
Therefore, $b_s^\top \epsilon_i^R$ is sub-Gaussian with variance proxy $\bar\sigma^2\|b\|_2^2$. Regarding $\|b_s\|_2$, we have
\[
\|b_s\|_2
= \big\|{\theta^{R}}({\theta^{R}}^\top \theta^R)^{-1}\vartheta_s\big\|_2
\ \le\ \|{\theta^{R}}\|_{\mathrm{op}}\;\|({\theta^{R}}^\top \theta^R)^{-1}\|_{\mathrm{op}}\;\|\vartheta_s\|_2.
\]
We note that $\lambda_{\min}({\theta^{R}}^\top \theta^R)\ge T\underline\xi$, leading to $\|({\theta^{R}}^\top \theta^R)^{-1}\|_{\mathrm{op}}\le (T\underline\xi)^{-1}$. Moreover,  $\|\theta^{R}\|_{\mathrm{op}}\le \sqrt{\sum_{t=1}^T\|\theta_t\|_2^2}\le \sqrt{T}\,\max_{t}\|\theta_t\|_2$, with $\|\theta_t\|_2\le \bar\theta$, and $\|\vartheta_s\|_2\le \bar\vartheta$. Therefore,
\[
\|b_s\|_2 \ \le\ \frac{\sqrt{T}\,\bar\theta}{T\underline\xi}\,\bar\vartheta
= \frac{\bar\vartheta\,\bar\theta}{\underline\xi}\,\frac{1}{\sqrt{T}},
\]
and consequently, $b_s^\top \epsilon_i^R$ is sub-Gaussian with variance proxy $(\bar\sigma\frac{\bar\vartheta\,\bar\theta}{\underline\xi}\,\frac{1}{\sqrt{T}})^2$.
Consequently, applying Theorem 1.16 from \citep{rigollet2023high}, we have
\begin{align*}
    |\E[R_{1s}]|
&=\left|\E\left[   \sum_{i=2}^{J+1} {w_i} \{b_s^\top \epsilon_i^R\}  \right] \right|    \\
&\le\E\left[\max_{w \in \mathbf{W}}   \left|\sum_{i=2}^{J+1} {w_i} \{b_s^\top \epsilon_i^R\}  \right| \right]    \\
&\le \sqrt{2}\bar{\vartheta} \bar{\theta} \bar{\sigma}\underline\xi^{-1} \sqrt{\frac{\log(2J)}{T}},
\end{align*}
where $\mathbf{W}$ is the simplex
\[
\mathbf{W}:=\left\{w\in\mathbb R^J:~ w_i\ge 0,\ \sum_{i=2}^{J+1}w_i=1\right\}.
\]

Similarly, we re-write $R_{2s}$ as
\begin{align*}
    R_{2s} &= -\vartheta_s^\top ({\theta^{R}}^\top \theta^R)^{-1}{\theta^{R}}^\top\tilde\theta^R (\tilde\mu^F_1 - \sum_{i=2}^{J+1} w_i \tilde\mu^F_i)\\
    &= \sum_{i=2}^{J+1} {w_i} \vartheta_s^\top ({\theta^{R}}^\top \theta^R)^{-1} {\theta^{R}}^\top \tilde\theta^R  (\tilde\mu_i^F-\tilde\mu_1^F)\\
    &= \sum_{i=2}^{J+1} {w_i} \tilde b_s^\top (\tilde\mu_i^F-\tilde\mu_1^F),
\end{align*}
where $\tilde b_s=\tilde\theta^{R^\top}{\theta^{R}}(\theta^{R^\top} \theta^R)^{-1}\vartheta_s$.
Notice that $(\tilde\mu_i^F-\tilde\mu_1^F)$ has independent mean-zero sub-Gaussian entries with variance proxy $2\tau^2$.
Therefore, $\tilde b_s^\top (\tilde\mu_i^F-\tilde\mu_1^F)$ is sub-Gaussian with variance proxy $2\tau^2\|\tilde b\|_2^2$. Regarding $\|\tilde b_s\|_2$, we have
\[
\|\tilde b_s\|_2
= \big\|\tilde\theta^{R^\top}{\theta^{R}}(\theta^{R^\top} \theta^R)^{-1}\vartheta_s\big\|_2
\ \le\ 
\|{\tilde\theta^R}\|_{\mathrm{op}}\;
\|{\theta^{R}}\|_{\mathrm{op}}\;\|(\theta^{R^\top} \theta^R)^{-1}\|_{\mathrm{op}}\;\|\vartheta_s\|_2.
\]
We note that $\lambda_{\min}({\theta^{R}}^\top \theta^R)\ge T\underline\xi$, leading to $\|({\theta^{R}}^\top \theta^R)^{-1}\|_{\mathrm{op}}\le (T\underline\xi)^{-1}$. Moreover,  
$\|\theta^{R}\|_{\mathrm{op}}\le \sqrt{\sum_{t=1}^T\|\theta_t\|_2^2}\le \sqrt{T}\,\max_{t}\|\theta_t\|_2$, with $\|\theta_t\|_2\le \bar\theta$, 
$\|\tilde\theta^{R}\|_{\mathrm{op}}\le \sqrt{\sum_{t=1}^T\|\tilde\theta_t\|_2^2}\le \sqrt{T}\,\max_{t}\|\tilde\theta_t\|_2$, with $\|\tilde\theta_t\|_2\le \bar{\tilde\theta}$, 
and $\|\vartheta_s\|_2\le \bar\vartheta$. Therefore,
\[
\|\tilde b_s\|_2 \ \le\ \frac{\sqrt{T}\,\bar\theta\sqrt{T}\,\bar{\tilde\theta}}{T\underline\xi}\,\bar\vartheta
= \frac{\bar\vartheta\,\bar\theta\bar{\tilde\theta}}{\underline\xi}\,
\]
and consequently, $\tilde b_s^\top (\tilde\mu_i^F-\tilde\mu_1^F)$ is sub-Gaussian with variance proxy $2(\tau\frac{\bar\vartheta\,\bar\theta\bar{\tilde\theta}}{\underline\xi})^2$.
Consequently, applying Theorem 1.16 from \citep{rigollet2023high}, we have
\begin{align*}
    |\E[R_{2s}]|
&=\left|\E\left[   \sum_{i=2}^{J+1} {w_i} \{\tilde b_s^\top (\tilde\mu_i^F-\tilde\mu_1^F)\}  \right] \right|    \\
&\le\E\left[\max_{w \in \mathbf{W}}   \left|\sum_{i=2}^{J+1} {w_i} \{\tilde b_s^\top (\tilde\mu_i^F-\tilde\mu_1^F)\}  \right| \right]    \\
&\le 2\bar\vartheta\,\bar\theta\bar{\tilde\theta}\underline\xi^{-1} \sqrt{\log(2J)},
\end{align*}
where $\mathbf{W}$ is the simplex
\[
\mathbf{W}:=\left\{w\in\mathbb R^J:~ w_i\ge 0,\ \sum_{i=2}^{J+1}w_i=1\right\}.
\]

Then
\[
|\mathbb E[R_{1s}]|
\le \sqrt{2}\,\frac{\bar\vartheta\,\bar\theta}{\underline\xi}\,\bar\sigma
\sqrt{\frac{\log(2J)}{T}},
\qquad
|\mathbb E[R_{2s}]|
\le 2\,\frac{\bar\vartheta\,\bar\theta\,\bar{\tilde\theta}}{\underline\xi}\,\tau\,\sqrt{\log(2J)}.
\]

\paragraph{Conclusion.}
Combining these bounds with \eqref{eq:triangle_new} and averaging over $s=1,\ldots,S$,
\begin{align*}
\big|\mathbb E[\hat\psi^{sc}] - \psi_0\big|
&= \left| \frac{1}{S} \sum_{s=1}^S \E\left[Y_{1,s} - \sum_{i=2}^{J+1} w_i Y_{i,s} \right] - \frac{1}{S} \sum_{s=1}^S \E \left[Y_{1,s} - Y_{1,s}^{(0)}\right] \right| \\
&=
\left|\frac{1}{S}\sum_{s=1}^S \mathbb E\left[Y_{1,s}^{(0)} - \sum_{i=2}^{J+1} w_i Y_{i,s}\right]\right| \\
& \le \frac{1}{S}\sum_{s=1}^S
\left|\mathbb E\left[Y_{1,s}^{(0)} - \sum_{i=2}^{J+1} w_i Y_{i,s}\right]\right|\\
&\le
\Bigg(
\bar{\varphi}\sqrt{d_t}
+
\frac{\bar\vartheta\,\bar\theta}{\underline\xi}
+
\frac{\bar\vartheta\,\bar\theta\,\bar\phi}{\underline\xi}\sqrt{d_r}
\Bigg)c
+
\sqrt{2}\,\frac{\bar\vartheta\,\bar\theta}{\underline\xi}\,\bar\sigma
\sqrt{\frac{\log(2J)}{T}}
+
2\,\frac{\bar\vartheta\,\bar\theta\,\bar{\tilde\theta}}{\underline\xi}\,\tau\,\sqrt{\log(2J)}.
\end{align*}
This proves the theorem.

\qed

\end{document}